\documentclass[aps,prl,10pt,twocolumn,superscriptaddress,notitlepage,showpacs,floatfix,longbibliography]{revtex4-2}

\usepackage{amsmath,amssymb,amsfonts}
\usepackage{graphicx}
\usepackage{dcolumn}
\usepackage{bm}
\usepackage{hyperref}
\usepackage{color}
\usepackage{physics}
\usepackage{mathtools}
\usepackage{xcolor}
\usepackage{amsthm}
\usepackage{comment}

\hypersetup{
    colorlinks=true,       
    linkcolor=red,          
   citecolor=magenta,        
    filecolor=magenta,      
    urlcolor=cyan,           
    runcolor=cyan
}

\newtheorem{definition}{Definition}
\newtheorem{lemma}{Lemma}
\newtheorem{proposition}{Proposition}

\newtheorem{result}{Result}

\newcommand{\Cg}{C_g}
\newcommand{\KM}{\ell_{\mathcal{M}}}

\newcommand{\KCg}{\ell_{C_g}}

\newcommand{\NM}{n_{\mathcal{M}}}

\newcommand{\NCg}{n_{C_g}}
\newcommand{\BN}{\mathcal{B}_N}

\renewcommand{\Pr}{\mathrm{Pr}}
\newcommand{\argmin}{\mathrm{arg\,min}}
\renewcommand{\Tr}{\mathrm{Tr}}

\begin{document}

\title{Minimal Subsystems for Characterizing Genuine Multipartite Entanglement
}

\author{Shashank Shekhar}
\email{Contact author: shashankshekhar01234@gmail.com}
\affiliation{Department of Physics, Indian Institute of Technology Bombay, Mumbai, Maharashtra 400076, India}
\affiliation{Networked Quantum Devices Unit, Okinawa Institute of Science and Technology Graduate University, Okinawa, Japan}

\author{Siddhartha Santra}
\email{Contact author: santra@iitb.ac.in}
\affiliation{Department of Physics, Indian Institute of Technology Bombay, Mumbai, Maharashtra 400076, India}
\affiliation{Center of Excellence in Quantum
Information, Computation, Science and Technology,
Indian Institute of Technology Bombay, Mumbai, Maharashtra 400076, India}

\begin{abstract}

Genuine multipartite entanglement (GME) is a valuable resource for quantum technologies, but quantifying it in many-body systems can be demanding since measures defined through optimizations over bipartitions are often assumed to require large joint measurements and many measurement settings. We introduce the concepts of GME evaluation length and GME evaluation number, which capture, respectively, the largest subsystem that must be jointly measured and the number of distinct subsystems that need to be measured to evaluate multipartite entanglement measures defined via bipartition optimizations. Taking the genuine multipartite concurrence (GMC) as a canonical example, we show that, for $N$-qubit Haar-random pure states, the GMC can be evaluated using only $N$ single-qubit purity measurements. Thus, for typical many-body states, evaluating the GMC requires considering only the $N$ single-site bipartitions, rather than the exponentially many distinct bipartitions. We further corroborate this behavior numerically for physically accessible states generated by local random quantum circuits. Our results show that the practical cost of quantifying GME can be dramatically lower than suggested by its formal definition, establishing measurement complexity as a complementary operational lens for multipartite entanglement.

\end{abstract}

\maketitle

\textit{Introduction--}Genuine multipartite entanglement (GME) characterizes correlations in a quantum
state that cannot be reproduced by any biseparable decomposition~\cite{horodecki2024multipartiteentanglement,walter2017multipartiteentanglement}.
Such correlations are fundamentally important in quantum many-body physics and underlie multipartite quantum information protocols including quantum communication ~\cite{Mark1999QSSintro, Epping2017NQKD, Hah2020AQCKAintro}, clock synchronization~\cite{Komar2014CSintro}, metrology \cite{GMEinQM1,GMEinQM2,GMEinQM3}, quantum computing \cite{GMEinQC1,GMEinQC2}.

Despite this centrality, the experimental evaluation of GME measures remains costly \cite{ent_detection,GME_eval_prob1,GME_eval_prob2}. Most measures are defined via an optimization over bipartitions: to evaluate them exactly, one must access reduced density matrices across all nontrivial bipartitions \cite{Ma2011,GGM,GMN,GMN_fermionic,GMN_fermionic2}. For an $N$-party system, the number of nontrivial bipartitions grows as $(2^{N-1}-1)$, and evaluating the GME generally requires joint measurements on subsystems containing up to $\lfloor N/2\rfloor$ qubits. Consequently, the experimental cost arises from two sources: the need to perform increasingly large joint measurements as the subsystem size grows, and the exponential growth in the number of distinct bipartitions that must be examined.

The exponential scaling above reflects a worst-case interpretation of the optimization defining the measure. In practice, however, the extremizing bipartition need not be arbitrary: it may be constrained by the structure of the state or the dynamics generating it and thus in realistic applications, one typically has prior information about the preparation circuit, governing dynamics, symmetries, Hamiltonian, or state family. This motivates the following operational question: \emph{For a given family of states and a given GME measure, what is the largest subsystem on which a joint measurement must be performed, and how many distinct subsystem measurements are actually required to evaluate the measure?}
This broad operational question has motivated extensive prior work \cite{EDL_motivation1,motivation2,motivation3,motivation4,motivation5,motivation6,Li2026}. Recent works in classical shadows \cite{motivation2} and few-copy certification \cite{motivation4} have drastically reduced the sample complexity of entanglement verification by leveraging randomized local probes and prior state knowledge. A parallel thread of work has sought to characterize the global state with marginals~\cite{marg1,marg2,marg3,marg4,marg5}. However, a rigorous framework that directly dictates the minimum operational joint measurement size and number of subsystem-level measurements required to evaluate GME is still lacking.



\begin{figure}
    \centering
    \includegraphics[width=\linewidth]{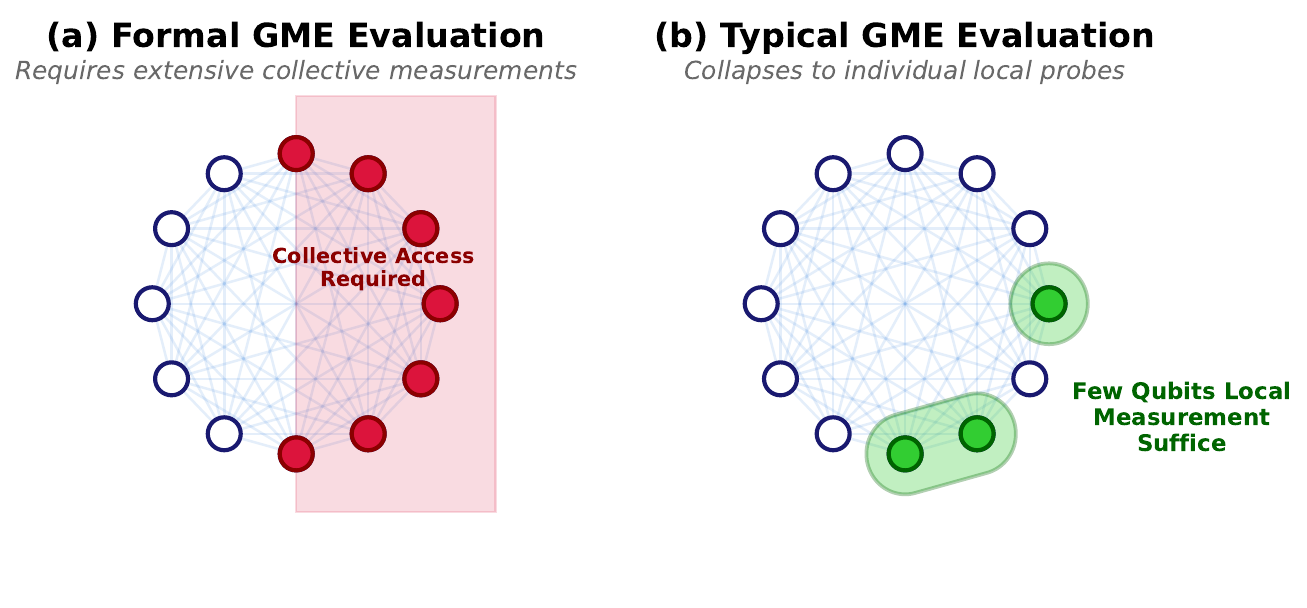}
    \caption{Collapse of measurement complexity for genuine multipartite entanglement. (a) \textbf{Formal evaluation:} Exact computation of bipartition-extremized measures such as the genuine multipartite concurrence (GMC) appears to require optimization over all $2^{N-1}-1$ non-trivial bipartitions, involving joint measurements on up to $\lfloor N/2 \rfloor$ qubits. (b) \textbf{Typical evaluation:} For typical highly entangled many-body states, the minimization collapses to the single-qubit sector, reducing the search to $N$ single-qubit measurements. The operational cost therefore scales linearly with $N$, rather than exponentially.}

    \label{fig:motivation}
\end{figure}

In this work, we address this by introducing the \emph{GME evaluation length} (GEL) $\KM(\mathcal F)$, which quantifies the size (in qubits) of the largest joint subsystem measurement needed to exactly evaluate a bipartition-extremized GME measure for every state in a family $\mathcal F$, and the \emph{GME evaluation number} (GEN) $\NM(\mathcal F)$, which counts the number of such distinct subsystems to be measured. For states drawn from a probabilistic ensemble $(\mathcal F,\mu)$, we define the corresponding $\varepsilon$-typical GME evaluation length $\ell_{\mathcal M}^{\varepsilon}(\mathcal F)$ and evaluation number $n_{\mathcal M}^{\varepsilon}(\mathcal F)$. We then establish several general properties of GEL, GEN, and $\varepsilon$-typical GEL and GEN, clarifying their operational meaning as measures of entanglement-evaluation complexity.

We instantiate this framework with the genuine multipartite concurrence (GMC) $C_g$~\cite{Ma2011}, defined as a minimization of bipartite concurrences over all bipartitions. We prove that, for $N$-qubit Haar-random pure states, $\ell_{C_g}^{\varepsilon}$ and $n_{C_g}^{\varepsilon}$ concentrate at $1$ and $N$, respectively, with probability $1-O(2^{-N})$. Thus, with overwhelming probability, the apparent exponential measurement complexity of GMC collapses to just $N$ possible single-site measurement settings for typical highly entangled states. We demonstrate numerically that this collapse also emerges dynamically in states generated by local random quantum circuits. Beyond quantifying measurement cost, GEL and GEN provide a natural framework for classifying multipartite states and ensembles by measurement complexity, while also inducing a partial ordering according to the resources required for evaluating GME.

\textit{Bipartition-extremized GME measures--}Many pure-state GME measures are obtained by evaluating a bipartite entanglement measure across every bipartition and extremizing over all cuts~\cite{Ma2011,GGM,GMentropy}, a paradigm that has also recently appeared in quantum spin liquids~\cite{GMN} and fermionic multipartite entanglement~\cite{GMN_fermionic}. This construction is natural because GME characterizes entanglement present across all bipartitions, making the extremal cut the relevant bottleneck. Moreover, the measure is determined by a single extremal bipartition rather than by an explicit aggregation over all cuts, potentially simplifying its evaluation.




Let $|\psi\rangle\in\bigotimes_{i=1}^N\mathcal H_i$ be a pure $N$-partite state, and let $\mathcal{B}_N$ denote the set of all nontrivial bipartitions, $\gamma=(A|B)$, of the $N$-parties. For each bipartition, define its size 
\(
|\gamma|:=\min(|A|,|B|),
\)
and let $\rho_\gamma$ denote the reduced density matrix of the smaller subsystem.

\begin{definition}[Bipartition-extremized measure]
A GME measure $\mathcal M$ is \emph{bipartition-extremized} if there exists a bipartite entanglement functional $f_\gamma$, depending only on $\rho_\gamma$, such that
\begin{equation}
  \mathcal{M}(|\psi\rangle) = \underset{\gamma \in \BN}{\mathrm{extr}}\;
  f_\gamma(|\psi\rangle),
  \label{eq:bpextreme}
\end{equation}
where $\operatorname{extr}$ denotes either a minimum or maximum, depending on the particular measure.
\end{definition}


The genuine multipartite concurrence (GMC)~\cite{Ma2011} provides a canonical example. For a bipartition $\gamma$ the bipartite concurrence,
\(
C_{A|B}(|\psi\rangle)=\sqrt{2\bigl(1-\Tr\rho_A^2\bigr)}
\),
with $\rho_A=\Tr_B|\psi\rangle\langle\psi|$, minimized over all bipartitions yields the GMC as,
\begin{equation}
\Cg(|\psi\rangle)=\min_{\gamma\in\BN} C_{A|B}(|\psi\rangle).
\label{eq:GMC}
\end{equation}
This measure is especially well suited to our purposes: it extends the familiar bipartite concurrence to the multipartite setting and depends only on subsystem purities, which are directly accessible from reduced measurements \cite{Purity_meas}. It is therefore a natural candidate for studying multipartite entanglement evaluation.



\textit{GME evaluation length and number---}While the formal definition of a GME measure implies an exponentially large measurement overhead, evaluating it on a family of pure states $\mathcal{F}$ of $N$-qubits with certain known properties, such as symmetries or generating dynamics, may require vastly fewer resources. To rigorously capture this operational reality, we introduce the GME evaluation length and evaluation number which formalize the measurement size in terms of the number of qubits, $|\gamma|$, and the number of bipartitions that must be sampled to evaluate a given bipartition-extremized GME measure $\mathcal{M}(\ket{\psi})$ for all $\ket{\psi}\in\mathcal{F}$.

\begin{definition}[GME Evaluation Length and Number]
\label{def:GELandGEN}
Let $\mathcal{M}(\ket{\psi})$ be a bipartition-extremized GME measure and let $\mathcal{F}$ be a family of $N$-qubit pure states. A subset of bipartitions $\mathcal{S}\subseteq\mathcal{B}_N$ is said to be \emph{$\mathcal{M}$-sufficient for $\mathcal{F}$} if,
\begin{equation}\label{eq:min_dem_set}
    \underset{\gamma\in\mathcal{S}}{\mathrm{extr}}\,
    f_\gamma(\ket{\psi})
    =
    \mathcal{M}(\ket{\psi}),
    \qquad
    \forall\,\ket{\psi}\in\mathcal{F}.
\end{equation}
\end{definition}

Among all $\mathcal{M}$-sufficient sets for $\mathcal{F}$, the
\emph{minimally demanding sufficient sets} are those that minimize the size of the largest bipartition they contain and, subject to that, their cardinality. These can be defined by,
\begin{equation}
\mathcal{S}^{*}(\mathcal{F})
=
\argmin_{\mathcal{S}\,:\,\mathcal{M}\text{-sufficient for }\mathcal{F}}
\left(
\max_{\gamma\in\mathcal{S}} |\gamma|,
\, |\mathcal{S}|
\right),
\label{eq:min_dem_suf_set}
\end{equation}
where the minimization is lexicographic, that is, $\max_{\gamma\in\mathcal{S}}|\gamma|$ is minimized first, and ties are broken by minimizing $|\mathcal{S}|$. The set $\mathcal{S}^{*}(\mathcal{F})$ comprises all $\mathcal{M}$-sufficient sets $\mathcal{S}_i(\mathcal{F})$, $i=1,2,\ldots$, that attain the optimum in Eq.~(\ref{eq:min_dem_suf_set}), i.e., all such sets have the same optimal values of $\max_{\gamma\in\mathcal{S}_i(\mathcal{F})}|\gamma|$ and
$|\mathcal{S}_i(\mathcal{F})|$.

The GEL and GEN for the family $\mathcal{F}$ can then be defined as the size of the largest bipartition in any $\mathcal{S}_i(\mathcal{F})\in\mathcal{S}^{*}(\mathcal{F})$ and its cardinality respectively,
\begin{align}
  \KM(\mathcal{F}) \;=\;
  \max_{\gamma \in \mathcal{S}_i(\mathcal{F})}\; |\gamma|, \qquad
\NM(\mathcal{F}) \;=\; |\mathcal{S}_i(\mathcal{F})|,
\label{eq:GELandGEN}
\end{align}
which, take positive integer values in the range, $1\leq \KM(\mathcal{F}) \leq \lfloor N/2\rfloor$ and 
$1\leq \NM(\mathcal{F})\leq \sum_{k=1}^{\KM(\mathcal{F})}\binom{N}{k}\leq (2^{N-1}-1)$. 
The minimum values, $\KM(\mathcal{F})=1$ and $\NM(\mathcal{F})=1$, imply that
only single-qubit measurements and a single bipartition, respectively, are
sufficient to evaluate the GME measure for every state in $\mathcal{F}$. This
is the case, for example, for family of $N$-qubit generalised GHZ state (Eq.~\ref{eq:gGHZ}). Conversely, $\KM(\mathcal{F})$ and $\NM(\mathcal{F})$
attain their maximum values when the full optimization over all bipartitions,
as in Eq.~(\ref{eq:bpextreme}), is required. Importantly, GEN does not count experimental repetitions, copies, or shots needed to estimate an individual subsystem quantity to a prescribed statistical precision. Therefore, the framework is also distinct from approaches that reduce the statistical sample complexity of entanglement estimation~\cite{motivation2,motivation3,motivation4}.

GEL and GEN can be defined even when $\mathcal M$ is not a bipartition-extremized quantity but instead admits a closed-form expression that can be reconstructed exactly from a fixed-size marginal $\rho_\gamma$ for some $\gamma$ of bounded size, by replacing the condition of Eq.~\eqref{eq:min_dem_set} with a re-constructibility condition $\mathcal M(\rho)=h_\gamma(\rho_\gamma)$. The formal definition, together with a worked example under this criterion, is provided in the supplemental material.

For a family of pure states, $\mathcal{F}$ equipped with a probability measure $\mu$, such as the Haar measure, $\varepsilon$-typical $\mathcal{M}$-sufficient sets can be defined by relaxing Eq. (\ref{eq:min_dem_set}) to require that the extremization restricted to the set obtains the GME measure with probability at least $(1-\varepsilon)$. This leads to corresponding $\varepsilon$-typical generalisations of the quantities defined in Eqs. (\ref{eq:min_dem_suf_set}) and (\ref{eq:GELandGEN}) as follows.

\begin{definition}[$\varepsilon$-Typical GEL and GEN]
\label{def:typical}
Let $\mathcal{F}$ denote an ensemble of $N$-qubit pure states distributed according to the probability measure $\mu$, and let $\mathcal{M}$ be a bipartition-extremized GME measure. A subset of bipartitions $\mathcal{S}\subseteq\mathcal{B}_N$ is called \emph{$\varepsilon$-typical $\mathcal{M}$-sufficient} for $\mathcal{F}$ if,
\begin{equation}
  \Pr_{|\psi\rangle\sim\mu}
  \left[
  \underset{\gamma\in\mathcal{S}}{\mathrm{extr}}\,
  f_\gamma(|\psi\rangle)
  =
  \mathcal{M}(|\psi\rangle)
  \right]
  \geq 1-\varepsilon.
  \label{eq:e_typ_set}
\end{equation}
\end{definition}

Then, the $\varepsilon$-typical minimally demanding sufficient sets, $\mathcal{S}^{*,\varepsilon}(\mathcal{F})$, comprise all $\varepsilon$-typical $\mathcal{M}$-sufficient sets $\mathcal{S}_i(\mathcal{F})$ that attain the lexicographic minimum in Eq.~(\ref{eq:min_dem_suf_set}). The corresponding
$\varepsilon$-typical GEL and GEN are given by,
\begin{align}
  \ell_{\mathcal{M}}^{\varepsilon}(\mathcal{F})
  &= \max_{\gamma\in\mathcal{S}_i(\mathcal{F})}|\gamma|,
  \qquad
  n_{\mathcal{M}}^{\varepsilon}(\mathcal{F})
  = |\mathcal{S}_i(\mathcal{F})|,
\end{align}
for any $\mathcal{S}_i(\mathcal{F})\in
\mathcal{S}^{*,\varepsilon}(\mathcal{F})$ which also provide upper bounds on the ensemble-averaged GEL and GEN. Specifically,
\begin{align}
\left\langle\ell_{\mathcal{M}}\right\rangle_{\mathcal{F}}
&\leq
(1-\varepsilon)\ell_{\mathcal{M}}^{\varepsilon}(\mathcal{F})
+\varepsilon\left\lfloor\frac{N}{2}\right\rfloor,
\nonumber\\
\left\langle n_{\mathcal{M}}\right\rangle_{\mathcal{F}}
&\leq
(1-\varepsilon)n_{\mathcal{M}}^{\varepsilon}(\mathcal{F})
+\varepsilon(2^{N-1}-1).
\end{align}
Thus, as $\varepsilon\to0$, the $\varepsilon$-typical GEL and GEN provide increasingly tight upper bounds on the expected measurement cost for states drawn from the ensemble. This ensemble-level perspective connects the notions of GEL and GEN to the statistical properties of many-body entanglement, which we investigate below for Haar-random states and states generated by local random circuits after describing their general properties.

\textit{Properties of the GEL and GEN--}The GEL and GEN, and the corresponding $\varepsilon$-typical quantities $\ell_{\mathcal M}^{\varepsilon}$ and $n_{\mathcal M}^{\varepsilon}$, satisfy the following properties, whose proofs are given in the
Supplemental Material.

\paragraph{(a) Local-unitary invariance.}
The GEL and GEN are invariant under local unitaries $U=U_1\otimes\cdots\otimes U_N$: since $\mathcal{M}(U\ket{\psi})=\mathcal{M}(\ket{\psi})$ for every $\ket\psi$, the sufficient sets $\mathcal S_i(\mathcal F)\in\mathcal S^*(\mathcal F)$ -- and hence $\KM,\NM$ -- are unchanged. Correspondingly, for any ensemble  $(\mathcal F,\mu)$, $\ell_{\mathcal M}^{\varepsilon}$ and $n_{\mathcal M}^{\varepsilon}$ are likewise invariant. In constrast, neither the exact nor typical quantities are invariant under SLOCC, which can alter $\mathcal M(\ket\psi)$ and the ordering of $f_\gamma(\ket\psi)$, changing the extremizing bipartition(s).

\paragraph{(b) Monotonicity under family enlargement.}
Let $\mathcal F$ is equipped with the probability measure $\mu$ and
$\mathcal F'\supseteq\mathcal F$ with the probability measure $\mu'$ with $\mu$ being the conditional measure induced by $\mu'$ on $\mathcal F$ -- i.e.\
$
\Pr_{\ket\psi\sim\mu}\big[\ket\psi\in A\big] =
\Pr_{\ket\psi'\sim\mu'}\big[\ket\psi'\in A \,\big|\, \ket\psi'\in\mathcal F\big],
$ for all  $A\subseteq\mathcal F$
then, writing $p_{\mathcal F}:=\Pr_{\ket\psi'\sim\mu'}[\ket\psi'\in\mathcal F]$, it follows that,
for every $\varepsilon\in[0,1]$,
\begin{equation}
\Big(\KM^{\epsilon}(\mathcal F),\,\NM^{\epsilon}(\mathcal F)\Big)
\ \preceq\
\Big(\KM^{\epsilon p_{\mathcal{F}}}(\mathcal F'),\,\NM^{\epsilon p_{\mathcal{F}}}(\mathcal F')\Big), \nonumber
\end{equation}
which reduces to the same-$\varepsilon$ statement $(\KM^{\epsilon}(\mathcal F),\,\NM^{\epsilon}(\mathcal F))
\ \preceq\
(\KM^{\epsilon}(\mathcal F'),\,\NM^{\epsilon}(\mathcal F'))$ when $p_{\mathcal F}=1$ (i.e.\ $\mathcal F'\setminus\mathcal F$ carries no $\mu'$-probability mass) and to the exact statement as $\varepsilon=0$ -- i.e. $(\KM(\mathcal F),\,\NM(\mathcal F))
\ \preceq\
(\KM(\mathcal F'),\,\NM(\mathcal F'))$.



\paragraph{(c) Union bound.}
Let $\mathcal F_1(\mathcal F_2$) be equipped with 
$\mu_1(\mu_2)$. Equip $\mathcal F_1\cup\mathcal F_2$ with the measure $\mu$ defined by the two-stage process: with probability
$p\in[0,1]$, sample $\ket{\psi}\in\mathcal F_1$ according to $\mu_1$; with probability $1-p$, sample $\ket{\phi}\in\mathcal F_2$ according to $\mu_2$, then for every
$\varepsilon\in[0,1]$,
\begin{equation}
\ell_{\mathcal M}^{\varepsilon}(\mathcal{F}_1\cup\mathcal{F}_2)
\le
\max\big\{\ell_{\mathcal M}^{\varepsilon}(\mathcal{F}_1),\,\ell_{\mathcal M}^{\varepsilon}(\mathcal{F}_2)\big\}. \nonumber
\end{equation}
For $\varepsilon=0$, this reduces to the corresponding exact GEL statement -- $\KM(\mathcal{F}_1\cup\mathcal{F}_2) \le \max\{\KM(\mathcal{F}_1),\KM(\mathcal{F}_2)\}$.


Together, these properties show that GEL and GEN, and the $\varepsilon$-typical GEL and GEN, behave as natural complexity measures for multipartite
entanglement evaluation.

\textit{Operational relevance of GEL and GEN--}
The quantities $\KM(\mathcal F)$ and $\NM(\mathcal F)$ characterize a state family $\mathcal F$ relative to a chosen bipartition-extremized GME measure $\mathcal M$. Their operational significance arises when the state-preparation procedure (e.g., a circuit, Hamiltonian, or dynamical protocol) is known, and hence the family $\mathcal F$ of possible output states is known, while the multipartite entanglement $\mathcal M(\ket\psi)$ produced in a given run is not.

As an example, suppose the preparation procedure generates the $N$-qubit generalized GHZ family \cite{gGHZ1,gGHZ2}
\begin{equation}
\label{eq:gGHZ}
\ket{g\mathrm{GHZ}_N}
=a\ket{0}^{\otimes N}+b\ket{1}^{\otimes N},
\end{equation}
so that
$\mathcal F=\{\ket{g\mathrm{GHZ}_N}:a,b\in\mathbb C,\ |a|^2+|b|^2=1\}$,
and let $\mathcal M=C_g$. Since all single-qubit marginals are identical and every bipartition $\gamma\in\mathcal B_N$ yields the same bipartite entanglement, the extremization is attained by any bipartition with $|\gamma|=1$. Thus, a singleton set $\mathcal S=\{\gamma\}$ is minimally $C_g$-sufficient (Definition~\ref{def:GELandGEN}), giving
\begin{equation}
\KM(\mathcal F)=1,\qquad
\NM(\mathcal F)=1.
\end{equation}
In contrast, without prior knowledge of the preparation procedure, an arbitrary $N$-qubit state generally requires consideration of all admissible bipartitions to identify the extremizing partition defining $C_g$.

We note that although Definitions~\ref{def:GELandGEN} and
\ref{def:typical} and their subsequent properties are formulated for pure states, they extend unchanged to any class of mixed states for which $\mathcal M$ remains bipartition-extremized as in Eq.~\eqref{eq:bpextreme}, upon replacing $\ket{\psi}\mapsto\rho$. Concrete examples include a class of mixed fermionic stabilizer states~\cite{GMN_fermionic}, frustration-free fermion ground states, Majorana dimer states, even-degree graph states, and Majorana fermion code states~\cite{GMN_fermionic2}. See the supplemental material for further details.

\textit{GEL and GEN for Haar-random states--}
Haar-random pure states~\cite{page,Mele2024introductiontohaar} provide a natural ensemble $(\mathcal F_{\mathrm{Haar}},\mu_{\mathrm{Haar}})$ for analyzing the typical behavior of $\ell_{\mathcal M}^{\varepsilon}(\mathcal F_{\mathrm{Haar}})$ and $n_{\mathcal M}^{\varepsilon}(\mathcal F_{\mathrm{Haar}})$, where $\mathcal F_{\mathrm{Haar}}$ is
equipped with the Haar measure $\mu_{\mathrm Haar}$. For pure states, the bipartite concurrence $f_\gamma(\ket\psi)=C_\gamma(\ket\psi)$ across any cut $\gamma\in\mathcal B_N$ is a monotonically increasing function of
$1-\Tr\rho_\gamma^2$, so the minimization over concurrences defining $C_g$ in Eq.~(\ref{eq:GMC}) is equivalent to a \emph{maximization} over subsystem purities $P_\gamma:=\Tr\rho_\gamma^2$.

For a Haar-random pure state on $\mathcal H_A\otimes\mathcal H_B$ with $\dim\mathcal H_A=d_A$, $\dim\mathcal H_B=d_B$, denote the mean and variance of $P_\gamma$ by $\overline P(d_A,d_B)$ and $\sigma^2(d_A,d_B)$ respectively, with explicit expressions given in the Supplemental Material and in~\cite{page}. In an $N$-qubit system, a bipartition $\gamma$ of size $|\gamma|$ corresponds to $d_A=2^{|\gamma|}$, $d_B=2^{N-|\gamma|}$; since the purity distribution depends on $\gamma$ only through $|\gamma|$, write
$\overline P_{|\gamma|}:=\overline P(2^{|\gamma|},2^{N-|\gamma|})$ and $\sigma_{|\gamma|}^2:=\sigma^2(2^{|\gamma|},2^{N-|\gamma|})$. The mean purities satisfy the strict ordering
\begin{equation}
\overline P_1 > \overline P_2 > \cdots > \overline P_{\lfloor N/2\rfloor},
\end{equation}
with $\overline P_{|\gamma|}=\overline P_{N-|\gamma|}$ by the symmetry of the bipartition. The most unbalanced cut ($|\gamma|=1$), requiring only a single-qubit measurement, has the largest mean purity. Since the variances $\sigma_{|\gamma|}^2$ are strongly suppressed, this cut is expected to maximize the purity with high probability. This intuition is formalized in the following proposition.

\begin{proposition}
\label{prop:1}
Let $\ket\psi\sim\mu_{\mathrm Haar}$ be drawn from the Haar measure on $N$ qubits, and for $\gamma\in\mathcal B_N$ with size $|\gamma|$ let \(
f_\gamma(\ket\psi) = C_{\gamma}(\ket\psi) = \sqrt{2\big(1-\Tr\rho_A^2\big)}
\)
denote the bipartite concurrence across the cut $\gamma$. There
exists a sequence $\varepsilon_N\to0$, decaying exponentially in $N$ for large N ($\varepsilon_N\sim \mathcal{O}(2^{-N})$), such
that
\begin{equation}
\Pr_{\ket\psi\sim\mu_{\mathrm H}}\!\left[
\underset{\gamma\,:\,|\gamma|=1}{\mathrm{min}}\ C_{\gamma}(\psi) = \ C_{g}(\psi)
\right]
\;\ge\; 1-\varepsilon_N,
\label{eq:prop1}
\end{equation}
i.e.\ with probability at least $1-\varepsilon_N$ the minimization of $C_g(\psi)$ over $\mathcal B_N$ is attained already at bipartition size $|\gamma|=1$.
\end{proposition}
The proof is deferred to the supplemental material.

By leveraging the $\varepsilon$-typical framework (Definition~\ref{def:typical}), we capture the complexity of the practical task of quantifying GME in typical highly entangled states. Although the extremizing cut is a particular single-qubit bipartition, an experimentalist without \emph{a priori} knowledge of which one, must measure the entire equivalent subspace. Proposition~\ref{prop:1} shows that, with probability
$1-\varepsilon_N$, the minimization is confined to
$\mathcal B_N^{(1)}:=\{\gamma\in\mathcal B_N:|\gamma|=1\}$, so the
$\varepsilon_N$-typical sufficient measurement set consists of exactly $N$ single-qubit measurements.

\begin{figure}
    \centering
    \includegraphics[width=0.75\linewidth]{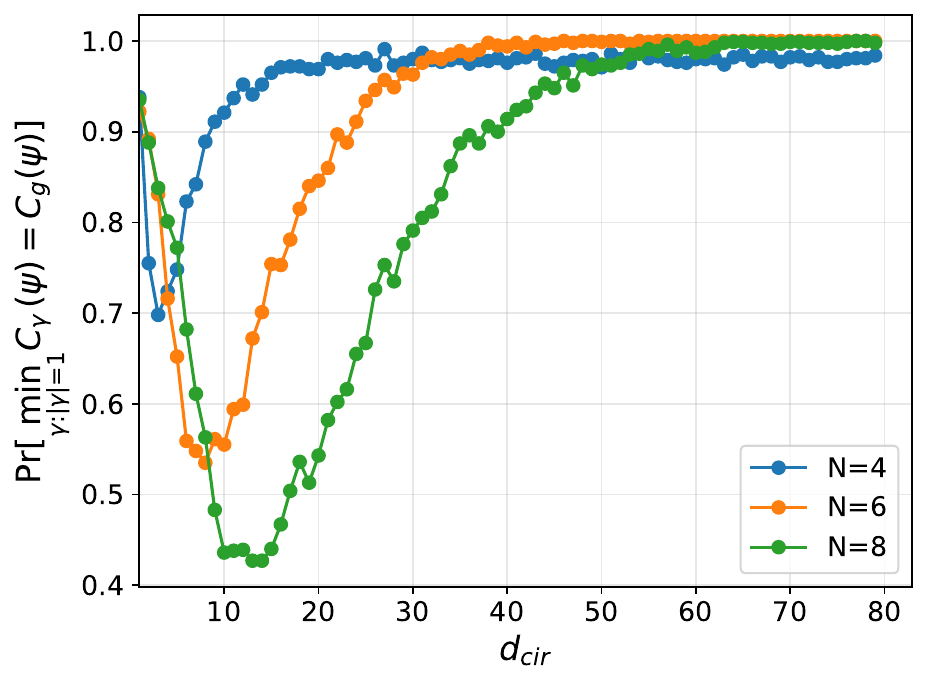}
    \caption{Behavior of the $C_g$ under local random circuits. The probability that the minimization defining $C_g$ is attained by the most unbalanced bipartition ($|\gamma|=1$) increases rapidly with $d_{\rm cir}$ and converges to unity for all system sizes considered. At small depths several cuts can compete due to the absence of GME, but this ambiguity disappears once scrambling develops. The minimization then becomes overwhelmingly dominated by the single-qubit cut.}
    \label{fig:cgmevsdepth2}
\end{figure}


\begin{result}
\label{prop:main}
Let $\mathcal F_{\mathrm{Haar}}$ denote the ensemble of $N$-qubit pure states equipped with the Haar measure $\mu_{\mathrm Haar}$, and let $\mathcal M=C_g$ be the GME measure. Then the set of all single-qubit bipartitions $\mathcal B_N^{(1)}$ is $\varepsilon_N$-typical $C_g$-sufficient for $\mathcal F_{\mathrm{Haar}}$, with $\varepsilon_N=O(2^{-N})$ as given in Proposition~\ref{prop:1}. Consequently, the $\varepsilon_N$-typical GEL and GEN attain the respective values,
\begin{equation}
\ell_{C_g}^{\varepsilon_N}(\mathcal F_{\mathrm{Haar}}) = 1,
\qquad
n_{C_g}^{\varepsilon_N}(\mathcal F_{\mathrm{Haar}}) = N.
\end{equation}
\end{result}

This drastically reduces the practical measurement cost relative to the worst-case bounds: the practical measurement size drops from
$\lfloor N/2\rfloor$ qubits to a single qubit ($\ell_{C_g}^{\varepsilon_N}(\mathcal F_{\mathrm{Haar}})=1$), and the worst-case number of measurements drops from the exponentially large $2^{N-1}-1$ to a strictly linear cost of $N$
($n_{C_g}^{\varepsilon_N}(\mathcal F_{\mathrm{Haar}})=N$).


Since $\ell_{C_g}^{\varepsilon_N}(\mathcal F_{\mathrm{Haar}}) = 1$, the average GMC follows from a perturbative expansion about the Haar-average single-qubit purity (see Supplemental Material):
\begin{equation}
\mathbb{E}[C_g]
\approx \sqrt{2(1-\overline P_1)}
-\frac{\sigma_1^2}{4\sqrt{2}\,(1-\overline P_1)^{3/2}},
\label{Eq:expected_Cg}
\end{equation}
implying $\mathbb{E}[C_g]\to1$ as $N\to\infty$.

\textit{Beyond Haar randomness: local random circuits--}
While Haar-random states provide a useful benchmark, physically realizable many-body states occupy only a restricted subset of Hilbert space. We consider the random-edge model of~\cite{QRC2_sid}. Starting from $\ket{00\cdots0}$, each circuit layer consists of selecting an adjacent pair $(i,i+1)$ uniformly at random from the $N-1$ edges of a one-dimensional chain and applying a Haar-random $U(4)$ gate. After $d_{\rm cir}$ such layers (the circuit depth), the resulting ensemble is denoted $\mathcal F_{\rm cir}(d_{\rm cir})$, with induced measure $\mu_{\rm cir}(d_{\rm cir})$.

\textit{Figure~\ref{fig:cgmevsdepth2} shows that
$\ell_{C_g}^{\varepsilon}\!\big(\mathcal F_{\rm cir}(d_{\rm cir})\big)=1$ and $n_{C_g}^{\varepsilon}\!\big(\mathcal F_{\rm cir}(d_{\rm cir})\big)=N$, with $\varepsilon\to0$ as the circuit depth $d_{\rm cir}$ and N increases.}

As local random circuits enter the volume-law regime, subsystem purities $P_\gamma$ concentrate around their Haar averages
$\overline P_k$~\cite{QRC1,QRC2_sid}. Since $C_g$ depends only on these purities, Proposition~\ref{prop:main} predicts that the minimization defining $C_g$ dynamically collapses onto the single-qubit subspace $\mathcal B_N^{(1)}$ as $\mathcal F_{\rm cir}(d_{\rm cir})$ approaches the Haar ensemble $\mathcal F_{\rm Haar}$. Thus, the effective reduction of GMC to a single-qubit quantity, $\ell_{C_g}^{\varepsilon}\!\big(\mathcal F_{\rm cir}(d_{\rm cir})\big)=1$ and $n_{C_g}^{\varepsilon}\!\big(\mathcal F_{\rm cir}(d_{\rm cir})\big)=N$, emerges dynamically under scrambling. See the supplemental material for further discussion.

\textit{Classification and partial ordering of GME states--}
Beyond quantifying the measurement cost of specific target states, GEL and GEN establish a formal operational taxonomy for multipartite entanglement. For a chosen $\mathcal M$, any state family $\mathcal F$ maps to a coordinate pair $\big(\KM(\mathcal F),\NM(\mathcal F)\big)$, defining equivalence classes based purely on measurement complexity. This structure naturally induces a strict partial order of experimental hardness. We say a state family $\mathcal F_A$ is strictly easier to evaluate than $\mathcal F_B$, denoted $\mathcal F_A \prec \mathcal F_B$, if $\KM(\mathcal F_A)\le \KM(\mathcal F_B)$ and $\NM(\mathcal F_A)\le \NM(\mathcal F_B)$, with at least one inequality strict. Families with incomparable coordinates under this ordering formalize experimental hardware trade-offs, such as the choice between executing a few macroscopic joint measurements versus a large ensemble of local probes.

Finally, this extends naturally to the statistical properties of
many-body systems. For an ensemble $(\mathcal F,\mu)$ with state preparation governed by the probability measure $\mu$, families can be classified by their $\varepsilon$-typical coordinates,
$\big(\ell_{\mathcal M}^{\varepsilon}(\mathcal F),\,n_{\mathcal
M}^{\varepsilon}(\mathcal F)\big)$, and the partial order $\prec$ extends directly to these coordinates for fixed $\varepsilon$. Here, $\varepsilon$ is an extra tunable operational parameter, allowing experimentalists to formally bound the trade-off between the strictness of the GME certification and the reduction of required measurement resources.

\textit{Summary and outlook--}We have introduced the concepts of \emph{GME evaluation length} and \emph{GME evaluation number}, operational quantities that characterize the measurement complexity associated with evaluating multipartite entanglement measures defined through optimizations over bipartitions. Using GMC on Haar-random states as a paradigmatic example, we showed that the measurement complexity of GME can be dramatically smaller than suggested by its formal definition. Beyond quantifying measurement cost, GEL and GEN induce a natural measurement-complexity classification of multipartite states and ensembles, together with a partial ordering based on the resources required for GME evaluation. 



More broadly, GEL and GEN provide a systematic way to exploit prior information about the relevant state family—through its preparation circuit, governing dynamics, symmetries, or Hamiltonian—to determine the measurement complexity of GME, which is crucial for quantum communication \cite{Mark1999QSSintro, Epping2017NQKD, Hah2020AQCKAintro}, quantum computation \cite{GMEinQC1,GMEinQC2}, and the characterization of correlated many-body systems \cite{QMB1,QMB2}. Future directions include extending the analysis of GEL and GEN to other physically relevant multipartite-state families—including ground states of local Hamiltonians, critical systems~\cite{Qcritical}, many-body localized phases~\cite{MBL1,MBL2}, Random hadamard states \cite{QHS1} and tensor-network states~\cite{MPstates}—as well as developing analogous notions of evaluation length and evaluation number for other quantum resources, such as coherence~\cite{coherence1,coherence2} and magic or nonstabilizerness~\cite{magic}. It may also guide the construction of quantitative GME witnesses~\cite{witness1,witness3,witness2} by identifying the bipartitions sufficient for certification with minimal measurement overhead. We thus expect GEL and GEN to provide a broadly applicable operational framework connecting multipartite-entanglement structure with the practical cost of its evaluation and certification.

\textit{Acknowledgements}: Funding from the IITB TRUST Labs grant DO/2023-SBST002-007 and the IITB seed funding is gratefully acknowledged.

\bibliography{ref.bib}

\clearpage


\onecolumngrid








\begin{center}
{\Large\bf Supplemental Material}\\[1em]
for\\[0.5em]
{\bf Minimal Subsystems for Characterizing Genuine Multipartite Entanglement}
\end{center}

\section{1. Brief Review of Haar-random states and genuine multipartite entanglement measures}
\subsection{A. Haar-Random States and Typicality}

We collect a few standard results for Haar-random pure states on a bipartite Hilbert space
$\mathcal{H}=\mathcal{H}_A\otimes\mathcal{H}_B$ with $\dim\mathcal{H}_A=m\le n=\dim\mathcal{H}_B$.

\paragraph{Average purity.}
The average purity of the subsystem $A$ (second moment of the reduced state) is
\begin{equation}
\mathbb{E}_{\psi}\big[\operatorname{Tr}\rho_A^2\big] \;=\; \frac{m+n}{mn+1}\,,
\label{eq:lubkin}
\end{equation}
a formula originally obtained by Lubkin. The variance is given by 
\begin{equation}
\operatorname{Var}\!\left[\operatorname{Tr}(\rho_A^2)\right]
=
\frac{2(m^2-1)(n^2-1)}
{(mn+1)^2(mn+2)(mn+3)}.
\end{equation}

\paragraph{Joint eigenvalue distribution.}
Let $\{\lambda_i\}_{i=1}^m$ be the eigenvalues of the reduced density matrix $\rho_A=\operatorname{Tr}_B|\psi\rangle\langle\psi|$.
The joint probability density of the spectrum under the Haar measure is
\begin{equation}
P(\{\lambda_i\}) \;=\; C_{m,n}\;\delta~\Big(1-\sum_{i=1}^m\lambda_i\Big)\,
\prod_{i=1}^m \lambda_i^{\,n-m}\prod_{1\le i<j\le m}(\lambda_i-\lambda_j)^2 ,
\end{equation}
where $C_{m,n}$ is the normalization constant and the $\delta$-function enforces unit trace.
This form (first derived by Lloyd and Pagels \cite{} and widely used thereafter) makes spectral averages tractable.

\paragraph{Average von Neumann entropy.}
For $m\le n$ the average von Neumann entropy of the subsystem $A$ under the Haar measure is given by Page's formula \cite{page}
\begin{equation}
\mathbb{E}_{\psi}\big[S(\rho_A)\big] \;=\; \sum_{j=n+1}^{mn} \frac{1}{j}\;-\;\frac{m-1}{2n}\,.
\label{eq:page}
\end{equation}
In the regime $n\gg m$ this reduces to the familiar approximation
\begin{equation}
\mathbb{E}[S(\rho_A)] \approx \ln m \;-\; \frac{m}{2n} \;+\; O~\!\big(n^{-2}\big),
\end{equation}
expressing the statement that the subsystem entropy is very close to its maximal value $\ln m$. Sen gave a compact rigorous proof of the exact Page expression \cite{siddharthasen}; several other independent proofs exist as well.


A parallel thread developed for higher moments and the full distribution of bipartite entanglement. Giraud \cite{Giraud2007} derived analytic expressions for the probability density of the purity (and linear entropy) for bipartite Haar-random pure states, computed all moments of the purity in closed form. These results complement the permutation/Weingarten-style moment formulae and have been used to obtain exact finite-\(N\) distributions for linear statistics of the Schmidt eigenvalues and related entanglement quantifiers. Other complementary approaches have emerged: Życzkowski and Sommers studied induced measures and provided many useful ensemble averages and moment formulae for reduced states \cite{ZyczkowskiSommers2001}; Coulomb-gas and random-matrix techniques produced large-\(N\) rate functions and revealed phase-transition behaviour in the distributions of Rényi entropies and purity; and several works computed distributional properties of concurrence- and determinant-based entanglement measures at finite \(N\). 

Further work has refined the notion of typicality by considering physically constrained ensembles and multipartite entanglement. Extensions beyond the Haar measure to ensembles generated by local Hamiltonians or shallow circuits characterize regimes in which Haar-like statistics persist or break down~\cite{QRC2_sid}. For multipartite settings, probabilistic approaches to entanglement in random states have been developed by Facchi \emph{et al.}~\cite{Facchi2006}, while Fitter \emph{et al.} have introduced algorithmic and probabilistic methods to estimate multipartite quantities such as the injective norm (equivalently, the geometric measure) for random states.

\subsection{B. Concurrence-Based Genuine Multipartite Entanglement}

Entanglement in many-body quantum systems can exhibit a genuinely multipartite structure, meaning it cannot be decomposed into correlations across any bipartition of the system. A pure $N$-partite state $|\psi\rangle\in\bigotimes_{i=1}^N\mathcal H_i$ is said to be genuinely multipartite entangled if it is inseparable across every bipartition. Since the number of nontrivial bipartitions grows exponentially with $N$, quantifying GME naturally involves a combinatorial optimization over bipartitions. This optimization constitutes the principal source of complexity for many GME measures.

The Genuine Multipartite Concurrence (GMC), introduced by Ma \emph{et al.}~\cite{Ma2011}, provides a natural extension of bipartite concurrence to multipartite systems by assigning the GME of a state according to its least entangled bipartition. For a bipartition $A|B$, let
\[
\rho_A=\operatorname{Tr}_B(|\psi\rangle\langle\psi|)
\]
denote the reduced density operator of subsystem $A$. The bipartite concurrence across the cut is

\begin{equation}
C_{A|B}(|\psi\rangle)
=
\sqrt{2\left(1-\operatorname{Tr}\rho_A^2\right)},
\label{eq:concurrence-purity}
\end{equation}

which depends only on the purity of the reduced state.

The GMC is defined as

\begin{equation}
C_g(|\psi\rangle)
=
\min_{A|B} C_{A|B}(|\psi\rangle).
\label{eq:GMC-def}
\end{equation}

By construction, $C_g(|\psi\rangle)>0$ if and only if $|\psi\rangle$ is genuinely multipartite entangled. The measure is invariant under local unitary transformations, is non-increasing under LOCC, and reduces to the standard concurrence for bipartite pure states. Its operational interpretation is particularly transparent: the value of the measure is determined entirely by the weakest entangled bipartition of the state.

Equation~(\ref{eq:GMC-def}) reveals that evaluating $C_g$ requires computing the reduced-state purity for every inequivalent bipartition before identifying the minimizing cut. Consequently, both the size of the required subsystem measurements and the exponentially large optimization over bipartitions contribute to the overall evaluation complexity. Understanding when this optimization can be substantially simplified is one of the central questions addressed in this work.

For mixed states, GMC is extended through the convex-roof construction.

\subsection{C. Generalized Geometric Measure}

Another widely used genuine multipartite entanglement measure is the Generalized Geometric Measure (GGM)~\cite{GGM}. Unlike GMC, which is expressed in terms of subsystem purities, GGM quantifies multipartite entanglement through the maximal overlap of a state with the set of biseparable states.

For a pure multipartite state $|\psi\rangle$, the GGM is defined as

\begin{equation}
\mathcal G(|\psi\rangle)
=
1-
\max_{|\phi\rangle\in\mathcal S_{\mathrm{bisep}}}
|\langle\phi|\psi\rangle|^2,
\label{eq:GGM-def}
\end{equation}

where $\mathcal S_{\mathrm{bisep}}$ denotes the set of all pure biseparable states.

A key simplification shows that the optimization over all biseparable states can be reduced to an optimization over bipartitions. Specifically,

\begin{equation}
\mathcal G(|\psi\rangle)
=
1-
\max_{A|B}
\lambda_{A|B,\max},
\label{eq:GGM-schmidt}
\end{equation}

where $\lambda_{A|B,\max}$ is the largest eigenvalue of the reduced density matrix $\rho_A$, or equivalently the square of the largest Schmidt coefficient across the bipartition $A|B$.

Like GMC, the GGM is a faithful measure of genuine multipartite entanglement: it vanishes exactly for biseparable states, is invariant under local unitary transformations, and is monotonic under LOCC for pure states~\cite{GGM}. Although its mathematical form differs from that of GMC, its computational structure is remarkably similar. Evaluating GGM likewise requires solving an optimization over all bipartitions of the system, with each bipartition contributing the largest Schmidt coefficient (or equivalently the largest eigenvalue of the reduced density operator). Consequently, the exponential growth in the number of bipartitions remains the dominant source of computational and experimental complexity.

\subsection{D. Fermionic Genuine Multipartite Negativity}

We now show that the fermionic genuine multipartite negativity (fGMN)~\cite{GMN_fermionic}~fits within the framework of bipartition-extremized measures. The relevant bipartite quantity is the \emph{fermionic negativity}, defined using the fermionic partial time-reversal operation across a bipartition.

For a state $\rho$ viewed across the bipartition $\gamma\mid\bar{\gamma}$, the fermionic negativity is
\begin{equation}
\mathcal{N}^{F}_{\gamma}(\rho)
=
\frac{1}{2}
\left\|
\rho^{R_\gamma}
\right\|_1
-\frac{1}{2}
=
\frac{1}{2}
\operatorname{tr}
\left|
\rho^{R_\gamma}
\right|
-\frac{1}{2},
\end{equation}
where $\|\cdot\|_1$ denotes the trace norm and $\rho^{R_\gamma}$ is the fermionic partial time-reversal of $\rho$ with respect to the subsystem $\gamma$.

The fermionic partial time-reversal acts on Majorana operators localized in $\gamma$ by multiplying each Majorana operator by a factor of $i$. Equivalently, when an operator is expanded in a basis of Majorana monomials, every Majorana operator belonging to $\gamma$ acquires this phase under $R_\gamma$. This modification is essential because the ordinary partial transpose fails to correctly capture fermionic entanglement due to exchange statistics and the fermion-parity superselection rule. The resulting fermionic negativity is an entanglement monotone under fermionic local operations and classical communication (fLOCC).

The connection with bipartition extremization follows from the inequality
\begin{equation}
\mathcal{N}_{F}(\rho)
\le
\min_{\gamma\in\mathcal{B}_N}
\mathcal{N}^{F}_{\gamma}(\rho),
\end{equation}
which shows that the fGMN is always upper bounded by the smallest fermionic negativity over all bipartitions. Thus, the weakest bipartition determines an upper bound on the genuine multipartite entanglement.

For pure states, this bound is saturated and the fermionic GMN is a bipartition-extremized measure in the sense of Definition~1. Explicitly,
\begin{equation}
\mathcal{N}_{F}\!\left(|\psi\rangle\langle\psi|\right)
=
\min_{\gamma\in\mathcal{B}_N}
f_\gamma(|\psi\rangle),
\end{equation}
where $f_\gamma(|\psi\rangle)=\mathcal{N}^{F}_{\gamma}(|\psi\rangle\langle\psi|)$.

The bipartition-extremized characterization also extends beyond pure states to an important class of mixed fermionic stabilizer states. These are states of the form
\begin{equation}
\rho
=
2^{-(N-S)}
\prod_{s=1}^{S}
\frac{1}{2}
\left(I+g_s\right),
\end{equation}
where the $\{g_s\}$ are commuting even-Majorana stabilizer generators. Such states are maximally mixed over the simultaneous $+1$ eigenspace of the stabilizer group and include reduced density matrices of pure fermionic stabilizer states, Majorana dimer states, even-degree graph states, and states arising in Majorana fermion codes.

For this entire stabilizer class,
\begin{equation}
\mathcal{N}_{F}(\rho)
=
\min_{\gamma\in\mathcal{B}_N}
\mathcal{N}^{F}_{\gamma}(\rho),
\end{equation}
so the fGMN remains bipartition-extremized even though the states are mixed.

\section{2. Proofs of Properties (a)--(c)}

Throughout, recall that $\mathcal S_i(\mathcal F)\in\mathcal S^*(\mathcal F)$
denotes any set attaining the lexicographic minimum of
$\big(\max_{\gamma\in\mathcal S}|\gamma|,|\mathcal S|\big)$ among sets
$\mathcal S$ satisfying Eq.~(3) for all
$\ket\psi\in\mathcal F$, and that
$\KM(\mathcal F)=\max_{\gamma\in\mathcal S_i(\mathcal F)}|\gamma|$,
$\NM(\mathcal F)=|\mathcal S_i(\mathcal F)|$. Similarly,
$\mathcal S_i(\mathcal F)\in\mathcal S^{*,\varepsilon}(\mathcal F)$ denotes
any set attaining the same lexicographic minimum among sets $\mathcal S$
satisfying the $\varepsilon$-typical condition
Eq.~(6), with
$\ell_{\mathcal M}^\varepsilon(\mathcal F)=\max_{\gamma\in\mathcal S_i(\mathcal F)}|\gamma|$,
$n_{\mathcal M}^\varepsilon(\mathcal F)=|\mathcal S_i(\mathcal F)|$.

We will use the following elementary fact repeatedly.

\begin{lemma}[Monotonicity in $\mathcal S$]
\label{lem:mono}
Let $\mathcal S_1\subseteq\mathcal S_2\subseteq\mathcal B_N$. If
\begin{equation}
\underset{\gamma\in\mathcal S_1}{\mathrm{extr}}\,f_\gamma(\ket\psi)=\mathcal M(\ket\psi)
\end{equation}
for a state $\ket\psi$, then
$\underset{\gamma\in\mathcal S_2}{\mathrm{extr}}\,f_\gamma(\ket\psi)=\mathcal M(\ket\psi)$
as well.
\end{lemma}
\begin{proof}
Take $\mathrm{extr}=\min$ (the case $\mathrm{extr}=\max$ is identical with
inequalities reversed). Since $\mathcal M(\ket\psi)=\min_{\gamma\in\mathcal
B_N}f_\gamma(\ket\psi)$, we have $f_\gamma(\ket\psi)\ge\mathcal M(\ket\psi)$
for every $\gamma\in\mathcal B_N$, and hence
$\min_{\gamma\in\mathcal T}f_\gamma(\ket\psi)\ge\mathcal M(\ket\psi)$ for any
$\mathcal T\subseteq\mathcal B_N$. By hypothesis,
$\min_{\gamma\in\mathcal S_1}f_\gamma(\ket\psi)=\mathcal M(\ket\psi)$. Since
$\mathcal S_1\subseteq\mathcal S_2$, minimizing over the larger set gives a
value no larger than minimizing over the smaller one:
\begin{equation}
\mathcal M(\ket\psi)\le\min_{\gamma\in\mathcal S_2}f_\gamma(\ket\psi)
\le\min_{\gamma\in\mathcal S_1}f_\gamma(\ket\psi)=\mathcal M(\ket\psi).
\end{equation}
Both inequalities are therefore equalities, proving the claim.
\end{proof}

\subsection*{Property (a): Local-unitary invariance}

\paragraph{(i) Proof for the exact quantities $\KM,\NM$.}
Let $U=U_1\otimes\cdots\otimes U_N$. For any bipartition
$\gamma=(A,\bar A)\in\mathcal B_N$, the reduced state of $U\ket\psi$ on $A$
is
\begin{equation}
\rho_A^{(U\psi)} = U_A\,\rho_A^{(\psi)}\,U_A^\dagger, \qquad U_A=\bigotimes_{i\in A}U_i,
\end{equation}
which is unitarily conjugate to $\rho_A^{(\psi)}$ and hence has the same
spectrum. Since $f_\gamma$ depends on $\ket\psi$ only through the spectrum
of $\rho_A^{(\psi)}$ (e.g.\ entanglement entropy, concurrence, etc.),
\begin{equation}
f_\gamma(U\ket\psi) = f_\gamma(\ket\psi) \qquad \forall\,\gamma\in\mathcal B_N,\ \forall\,\ket\psi.
\end{equation}
Consequently, for \emph{any} fixed subset $\mathcal S\subseteq\mathcal B_N$,
\begin{equation}
\underset{\gamma\in\mathcal S}{\mathrm{extr}}\,f_\gamma(U\ket\psi)
=\underset{\gamma\in\mathcal S}{\mathrm{extr}}\,f_\gamma(\ket\psi),
\qquad
\mathcal M(U\ket\psi)=\mathcal M(\ket\psi).
\end{equation}
Hence Eq.~(3) holds for $\mathcal S$ at $U\ket\psi$ if and
only if it holds at $\ket\psi$. Therefore $\mathcal S$ is $\mathcal
M$-sufficient for $U\mathcal F=\{U\ket\psi:\ket\psi\in\mathcal F\}$ if and
only if it is $\mathcal M$-sufficient for $\mathcal F$: the two families have
\emph{exactly the same set of $\mathcal M$-sufficient subsets}
$\mathcal S\subseteq\mathcal B_N$. Since $\mathcal S^*(\mathcal F)$ and
$\mathcal S^*(U\mathcal F)$ are obtained by the same lexicographic
minimization over identical feasible sets, they coincide, and in particular
\begin{equation}
\KM(U\mathcal F)=\KM(\mathcal F), \qquad \NM(U\mathcal F)=\NM(\mathcal F).
\end{equation}

\paragraph{(ii) Proof for the $\varepsilon$-typical quantities $\ell_{\mathcal M}^\varepsilon,n_{\mathcal M}^\varepsilon$.}
Let $(\mathcal F,\mu)$ be an ensemble and let $\mu_U$ denote the pushforward
of $\mu$ under $\ket\psi\mapsto U\ket\psi$, i.e.\ the distribution of
$U\ket\psi$ when $\ket\psi\sim\mu$ ($\mu_U(U\mathcal{F}):=\mu(U^{-1}U\mathcal{F})=Pr_{\ket{\psi} \sim \mu}[U\ket{\psi} \in U\mathcal{F}]$); this is the natural ensemble measure on
$U\mathcal F$. As shown above, for every fixed $\mathcal S\subseteq\mathcal
B_N$ the condition
$\mathrm{extr}_{\gamma\in\mathcal S}f_\gamma=\mathcal M$ holds at $U\ket\psi$
exactly when it holds at $\ket\psi$. Hence the two events have equal
probability:
\begin{equation}
\Pr_{\ket\phi\sim\mu_U}\!\Big[\underset{\gamma\in\mathcal S}{\mathrm{extr}}\,f_\gamma(\ket\phi)=\mathcal M(\ket\phi)\Big]
=
\Pr_{\ket\psi\sim\mu}\!\Big[\underset{\gamma\in\mathcal S}{\mathrm{extr}}\,f_\gamma(\ket\psi)=\mathcal M(\ket\psi)\Big],
\end{equation}
for every $\mathcal S\subseteq\mathcal B_N$. Thus $\mathcal S$ satisfies the
$\varepsilon$-typical sufficiency condition Eq.~(6) for
$(U\mathcal F,\mu_U)$ if and only if it satisfies it for $(\mathcal F,\mu)$:
the family of $\varepsilon$-typical $\mathcal M$-sufficient sets is
identical for both ensembles, for every $\varepsilon$. Consequently
$\mathcal S^{*,\varepsilon}(U\mathcal F)=\mathcal S^{*,\varepsilon}(\mathcal
F)$, and
\begin{equation}
\ell_{\mathcal M}^{\varepsilon}(U\mathcal F)=\ell_{\mathcal M}^{\varepsilon}(\mathcal F),
\qquad
n_{\mathcal M}^{\varepsilon}(U\mathcal F)=n_{\mathcal M}^{\varepsilon}(\mathcal F).
\end{equation}
In particular, if $\mu$ is itself LU-invariant (e.g.\ Haar measure), then
$\mu_U=\mu$ and $U\mathcal F=\mathcal F$, so the statement is immediate.

\subsection*{Property (b): Monotonicity under family enlargement}

\paragraph{(i) Proof for the exact quantities $\KM,\NM$.}
Let $\mathcal F\subseteq\mathcal F'$, and let
$\mathcal S_i(\mathcal F')\in\mathcal S^*(\mathcal F')$ be an optimal
sufficient set for $\mathcal F'$, so that by definition
\begin{equation}
\underset{\gamma\in\mathcal S_i(\mathcal F')}{\mathrm{extr}}\,f_\gamma(\ket\psi)=\mathcal M(\ket\psi)
\qquad \forall\,\ket\psi\in\mathcal F'.
\end{equation}
Since $\mathcal F\subseteq\mathcal F'$, this equality holds in particular for
every $\ket\psi\in\mathcal F$, so $\mathcal S_i(\mathcal F')$ is also
$\mathcal M$-sufficient for $\mathcal F$. Because
$\KM(\mathcal F)=\max_{\mathcal S:\text{ suff.\ for }\mathcal F}$
is a minimum (lexicographically) over a set that now includes
$\mathcal S_i(\mathcal F')$ as a feasible candidate,
\begin{equation}
\Big(\KM(\mathcal F),\,\NM(\mathcal F)\Big)
\ \preceq\
\Big(\max_{\gamma\in\mathcal S_i(\mathcal F')}|\gamma|,\ |\mathcal S_i(\mathcal F')|\Big)
=
\Big(\KM(\mathcal F'),\,\NM(\mathcal F')\Big),
\end{equation}
where $\preceq$ is the lexicographic order. In particular the first
coordinate gives $\KM(\mathcal F)\le \KM(\mathcal F')$; and if
$\KM(\mathcal F)=\KM(\mathcal F')$, minimality of $\mathcal S^*(\mathcal F)$
among sets with that fixed largest-bipartition size forces
$\NM(\mathcal F)\le \NM(\mathcal F')$. Hence
\begin{equation}
\KM(\mathcal F)\le\KM(\mathcal F'), \qquad \NM(\mathcal F)\le\NM(\mathcal F').
\end{equation}

\paragraph{(ii) Proof for the $\varepsilon$-typical quantities $\ell_{\mathcal M}^\varepsilon,n_{\mathcal M}^\varepsilon$.}

Let $\mathcal F\subseteq\mathcal F'$, with $\mathcal F$ equipped with the
probability measure $\mu$ and $\mathcal F'$ equipped with the probability
measure $\mu'$. We assume $\mu$ and $\mu'$ describe the \emph{same}
underlying random state-preparation process, with $\mu$ simply being that
process conditioned on the (now known) extra information that the
prepared state lies in $\mathcal F$. Precisely, for a state $\ket\psi'\sim\mu'$
drawn from the larger ensemble, and any subset $A\subseteq\mathcal F$, this
means
\begin{equation}
\Pr_{\ket\psi\sim\mu}\big[\ket\psi\in A\big]
\;=\;
\Pr_{\ket\psi'\sim\mu'}\big[\ket\psi'\in A \,\big|\, \ket\psi'\in\mathcal F\big]
\;=\;
\frac{\Pr_{\ket\psi'\sim\mu'}\big[\ket\psi'\in A\big]}{\Pr_{\ket\psi'\sim\mu'}\big[\ket\psi'\in \mathcal F\big]},
\label{eq:cond_b_prob}
\end{equation}
which is just the ordinary definition of conditional probability, applied
to the event $\{\ket\psi'\in\mathcal F\}$. Intuitively: if you draw
repeatedly from $\mu'$ but only keep the draws that happen to land in
$\mathcal F$, the resulting distribution over $\mathcal F$ \emph{is} $\mu$.

Let $\mathcal S_i(\mathcal F')\in\mathcal S^{*,\varepsilon'}(\mathcal F')$ be
an optimal set achieving $\ell_{\mathcal M}^{\varepsilon'}(\mathcal F')$, so
that by definition
\begin{equation}
\Pr_{\ket\psi'\sim\mu'}\!\left[\underset{\gamma\in\mathcal S_i(\mathcal F')}{\mathrm{extr}}\,f_\gamma(\ket\psi')\ne\mathcal M(\ket\psi')\right]\le\varepsilon'.
\label{eq:Fprime_fail}
\end{equation}

We want to bound the analogous failure probability, but now with $\ket\psi$
drawn from $\mu$ (i.e.\ restricted to $\mathcal F$) instead of from $\mu'$.
Since $\mathcal F\subseteq\mathcal F'$, the event
$$
\left\{\ket\psi\in\mathcal F : \underset{\gamma\in\mathcal S_i(\mathcal F')}{\mathrm{extr}}\,f_\gamma(\ket\psi)\ne\mathcal M(\ket\psi)\right\}
$$
is a well-defined subset $A\subseteq\mathcal F$, so Eq.~\eqref{eq:cond_b_prob}
applies directly to it:
\begin{align}
\Pr_{\ket\psi\sim\mu}\!\left[\underset{\gamma\in\mathcal S_i(\mathcal F')}{\mathrm{extr}}\,f_\gamma(\ket\psi)\ne\mathcal M(\ket\psi)\right]
&=
\frac{\displaystyle\Pr_{\ket\psi'\sim\mu'}\!\left[\ket\psi'\in\mathcal F \ \text{and}\ \underset{\gamma\in\mathcal S_i(\mathcal F')}{\mathrm{extr}}\,f_\gamma(\ket\psi')\ne\mathcal M(\ket\psi')\right]}
{\Pr_{\ket\psi'\sim\mu'}\big[\ket\psi'\in\mathcal F\big]} \label{eq:step1}\\[4pt]
&\le
\frac{\displaystyle\Pr_{\ket\psi'\sim\mu'}\!\left[\underset{\gamma\in\mathcal S_i(\mathcal F')}{\mathrm{extr}}\,f_\gamma(\ket\psi')\ne\mathcal M(\ket\psi')\right]}
{\Pr_{\ket\psi'\sim\mu'}\big[\ket\psi'\in\mathcal F\big]} \label{eq:step2}\\[4pt]
&\le
\frac{\varepsilon'}{\Pr_{\ket\psi'\sim\mu'}\big[\ket\psi'\in\mathcal F\big]}. \label{eq:step3}
\end{align}
Here, Eq.~\eqref{eq:step1} is Eq.~\eqref{eq:cond_b_prob} applied to the
failure event; Eq.~\eqref{eq:step2} holds because the event
"$\ket\psi'\in\mathcal F$ \emph{and} the extremization fails" is a
\emph{more restrictive} (smaller-or-equal-probability) event than simply
"the extremization fails," since we have dropped the requirement
$\ket\psi'\in\mathcal F$ and are now considering failure over the whole of
$\mathcal F'$; and Eq.~\eqref{eq:step3} substitutes the bound
\eqref{eq:Fprime_fail}.

Writing $p_{\mathcal F}:=\Pr_{\ket\psi'\sim\mu'}[\ket\psi'\in\mathcal F]$ for
brevity, Eq.~\eqref{eq:step3} says
\begin{equation}
\Pr_{\ket\psi\sim\mu}\!\left[\underset{\gamma\in\mathcal S_i(\mathcal F')}{\mathrm{extr}}\,f_\gamma(\ket\psi)\ne\mathcal M(\ket\psi)\right]
\le \frac{\varepsilon'}{p_{\mathcal F}},
\end{equation}
equivalently,
\begin{equation}
\Pr_{\ket\psi\sim\mu}\!\left[\underset{\gamma\in\mathcal S_i(\mathcal F')}{\mathrm{extr}}\,f_\gamma(\ket\psi)=\mathcal M(\ket\psi)\right]
\ge 1-\frac{\varepsilon'}{p_{\mathcal F}}.
\end{equation}
That is, $\mathcal S_i(\mathcal F')$ is $\big(\varepsilon'/p_{\mathcal
F}\big)$-typical $\mathcal M$-sufficient for $(\mathcal F,\mu)$, in the
sense of Definition~3.

Setting $\varepsilon := \varepsilon'/p_{\mathcal F}$ (equivalently
$\varepsilon' = \varepsilon\, p_{\mathcal F}$), the set $\mathcal
S_i(\mathcal F')$ is a \emph{feasible} candidate for the lexicographic
minimization defining $\mathcal S^{*,\varepsilon}(\mathcal F)$. By the same
lexicographic-minimality argument used in the exact case,
$\big(\ell_{\mathcal M}^{\varepsilon}(\mathcal F),\,n_{\mathcal
M}^{\varepsilon}(\mathcal F)\big)$ can only be smaller-or-equal (in the
lexicographic order) to the value attained by this candidate, giving

\begin{equation}
\Big(\KM^{\epsilon}(\mathcal F),\,\NM^{\epsilon}(\mathcal F)\Big)
\ \preceq\
\Big(\KM^{\epsilon p_{\mathcal{F}}}(\mathcal F'),\,\NM^{\epsilon p_{\mathcal{F}}}(\mathcal F')\Big),
\label{eq:final_b_typical}
\end{equation}

and in particular the first coordinate gives $\KM^{\epsilon}(\mathcal F) \preceq \KM^{\epsilon p_{\mathcal{F}}}(\mathcal F')$ and if $\KM^{\epsilon}(\mathcal F) = \KM^{\epsilon p_{\mathcal{F}}}(\mathcal F')$ then, \(
n_{\mathcal M}^{\varepsilon}(\mathcal F)\;\le\;n_{\mathcal M}^{\varepsilon\,p_{\mathcal F}}(\mathcal F'),
\quad\text{where } p_{\mathcal F} := \Pr_{\ket\psi'\sim\mu'}\big[\ket\psi'\in\mathcal F\big].
\)

\medskip
\noindent\textbf{Special cases.}
\begin{itemize}
\item As $\varepsilon\to0$, Eq.~\eqref{eq:final_b_typical} recovers the
exact-case statement $\KM(\mathcal F)\le\KM(\mathcal F')$,
$\NM(\mathcal F)\le\NM(\mathcal F')$.
\item If $p_{\mathcal F}=\Pr_{\ket\psi'\sim\mu'}[\ket\psi'\in\mathcal F]=1$
(i.e.\ $\mathcal F'\setminus\mathcal F$ carries no probability mass under
$\mu'$, so $\mu=\mu'$ for all practical purposes), Eq.~\eqref{eq:final_b_typical}
reduces to the clean same-$\varepsilon$ statement.
\end{itemize}

\medskip
\noindent\textbf{Why the rescaling by $p_{\mathcal F}$ is necessary.}
The clean same-$\varepsilon$ statement
$\ell_{\mathcal M}^{\varepsilon}(\mathcal F)\le\ell_{\mathcal M}^{\varepsilon}(\mathcal F')$
by direct analogy with the exact monotonicity property is false
in general: Eq.~\eqref{eq:final_b_typical} shows that the tighter
tolerance $\varepsilon\,p_{\mathcal F}$ on $\mathcal F'$ is what is actually
needed to guarantee tolerance $\varepsilon$ on the sub-ensemble $\mathcal F$,
and this rescaling cannot be removed. The reason is that
$\varepsilon'$-typical sufficiency for $\mathcal F'$ only bounds the
\emph{total} failure probability over $\mathcal F'$; it gives no control over
\emph{where} that failure mass is distributed. In the worst case, all of the
$\varepsilon'$ allowed failure probability could be concentrated entirely on
the rare sub-ensemble $\mathcal F$ (e.g.\ if $\mathcal F'\setminus\mathcal F$
consists of states on which the optimal set for $\mathcal F'$ always
succeeds, while it fails on every state of $\mathcal F$), in which case the
restricted failure probability on $\mathcal F$ equals $\varepsilon'/p_{\mathcal
F}$, saturating the bound in Eq.~\eqref{eq:step3}. Only when $p_{\mathcal
F}=1$ -- i.e.\ $\mathcal F'\setminus\mathcal F$ carries no probability mass
under $\mu'$, so there is no "elsewhere" for failures to accumulate -- does
the bound collapse to the clean same-$\varepsilon$ statement. In the
deterministic limit $\varepsilon\to0$, however, this subtlety disappears:
zero failure probability on $\mathcal F'$ forces zero failure probability on
every subset, including $\mathcal F$, recovering the exact monotonicity
property with no rescaling.

\subsection*{Property (c): Union bound}

\paragraph{(i) Proof for the exact quantity $\KM$ (and $\NM$).}
Let $\mathcal S_i(\mathcal F_1)\in\mathcal S^*(\mathcal F_1)$ and
$\mathcal S_j(\mathcal F_2)\in\mathcal S^*(\mathcal F_2)$ be optimal
sufficient sets, so
\begin{equation}
\underset{\gamma\in\mathcal S_i(\mathcal F_1)}{\mathrm{extr}}\,f_\gamma(\ket\psi)=\mathcal M(\ket\psi)\ \ \forall\,\ket\psi\in\mathcal F_1,
\qquad
\underset{\gamma\in\mathcal S_j(\mathcal F_2)}{\mathrm{extr}}\,f_\gamma(\ket\psi)=\mathcal M(\ket\psi)\ \ \forall\,\ket\psi\in\mathcal F_2.
\end{equation}
Set $\mathcal S:=\mathcal S_i(\mathcal F_1)\cup \mathcal S_j(\mathcal F_2)$.
Since $\mathcal S_i(\mathcal F_1)\subseteq\mathcal S$, Lemma~\ref{lem:mono}
gives $\mathrm{extr}_{\gamma\in\mathcal S}f_\gamma(\ket\psi)=\mathcal
M(\ket\psi)$ for every $\ket\psi\in\mathcal F_1$; likewise using
$\mathcal S_j(\mathcal F_2)\subseteq\mathcal S$, the same holds for every
$\ket\psi\in\mathcal F_2$. Hence $\mathcal S$ is $\mathcal M$-sufficient for
$\mathcal F_1\cup\mathcal F_2$, with
\begin{equation}
\max_{\gamma\in\mathcal S}|\gamma|
=\max\Big\{\max_{\gamma\in\mathcal S_i(\mathcal F_1)}|\gamma|,\ \max_{\gamma\in\mathcal S_j(\mathcal F_2)}|\gamma|\Big\}
=\max\{\KM(\mathcal F_1),\KM(\mathcal F_2)\}.
\end{equation}
Since $\KM(\mathcal F_1\cup\mathcal F_2)$ is by definition the minimal
achievable value of $\max_{\gamma\in\mathcal S}|\gamma|$ over $\mathcal
M$-sufficient sets for $\mathcal F_1\cup\mathcal F_2$, and $\mathcal S$ is
one explicit feasible choice,
\begin{equation}
\KM(\mathcal F_1\cup\mathcal F_2)\le \max\{\KM(\mathcal F_1),\KM(\mathcal F_2)\}.
\end{equation}

\paragraph{(ii) Proof for the $\varepsilon$-typical quantities $\ell_{\mathcal M}^\varepsilon,n_{\mathcal M}^\varepsilon$.}

Let $\mathcal F_1,\mathcal F_2$ be equipped with probability measures
$\mu_1,\mu_2$. Equip $\mathcal F_1\cup\mathcal F_2$ with the measure $\mu$
defined by the two-stage process: with probability $p\in[0,1]$ sample
$\ket\psi\sim\mu_1$, and with probability $1-p$ sample $\ket\psi\sim\mu_2$,
so that for any $A\subseteq\mathcal F_1\cup\mathcal F_2$,
\begin{equation}
\Pr_{\ket\psi\sim\mu}\big[\ket\psi\in A\big]
=
p\,\Pr_{\ket\psi\sim\mu_1}\big[\ket\psi\in A\cap\mathcal F_1\big]
+(1-p)\,\Pr_{\ket\psi\sim\mu_2}\big[\ket\psi\in A\cap\mathcal F_2\big].
\label{eq:mix_prob2}
\end{equation}

Let $\mathcal S_i(\mathcal F_1)\in\mathcal S^{*,\varepsilon}(\mathcal F_1)$
and $\mathcal S_j(\mathcal F_2)\in\mathcal S^{*,\varepsilon}(\mathcal F_2)$ be
optimal $\varepsilon$-typical sufficient sets, so
\begin{equation}
\Pr_{\ket\psi\sim\mu_1}\!\Big[\underset{\gamma\in\mathcal S_i(\mathcal F_1)}{\mathrm{extr}}f_\gamma(\ket\psi)=\mathcal M(\ket\psi)\Big]\ge1-\varepsilon,
\qquad
\Pr_{\ket\psi\sim\mu_2}\!\Big[\underset{\gamma\in\mathcal S_j(\mathcal F_2)}{\mathrm{extr}}f_\gamma(\ket\psi)=\mathcal M(\ket\psi)\Big]\ge1-\varepsilon.
\end{equation}
Set $\mathcal S:=\mathcal S_i(\mathcal F_1)\cup\mathcal S_j(\mathcal F_2)$. By
Lemma~\ref{lem:mono} (using $\mathcal S_i(\mathcal F_1)\subseteq\mathcal S$
and $\mathcal S_j(\mathcal F_2)\subseteq\mathcal S$),
\begin{equation}
\Pr_{\ket\psi\sim\mu_1}\!\Big[\underset{\gamma\in\mathcal S}{\mathrm{extr}}f_\gamma(\ket\psi)=\mathcal M(\ket\psi)\Big]\ge1-\varepsilon,
\qquad
\Pr_{\ket\psi\sim\mu_2}\!\Big[\underset{\gamma\in\mathcal S}{\mathrm{extr}}f_\gamma(\ket\psi)=\mathcal M(\ket\psi)\Big]\ge1-\varepsilon.
\end{equation}
Applying \eqref{eq:mix_prob2} to the success event of $\mathcal S$,
\begin{align}
\Pr_{\ket\psi\sim\mu}\!\Big[\underset{\gamma\in\mathcal S}{\mathrm{extr}}f_\gamma(\ket\psi)=\mathcal M(\ket\psi)\Big]
&= p\,\Pr_{\ket\psi\sim\mu_1}\!\Big[\underset{\gamma\in\mathcal S}{\mathrm{extr}}f_\gamma(\ket\psi)=\mathcal M(\ket\psi)\Big]
+(1-p)\,\Pr_{\ket\psi\sim\mu_2}\!\Big[\underset{\gamma\in\mathcal S}{\mathrm{extr}}f_\gamma(\ket\psi)=\mathcal M(\ket\psi)\Big] \nonumber\\
&\ge p(1-\varepsilon)+(1-p)(1-\varepsilon) = 1-\varepsilon,
\end{align}
for \emph{every} $p\in[0,1]$. Hence $\mathcal S$ is $\varepsilon$-typical
$\mathcal M$-sufficient for $(\mathcal F_1\cup\mathcal F_2,\mu)$, with
\begin{equation}
\max_{\gamma\in\mathcal S}|\gamma| = \max\big\{\ell_{\mathcal M}^\varepsilon(\mathcal F_1),\ell_{\mathcal M}^\varepsilon(\mathcal F_2)\big\},
\end{equation}
and since $\ell_{\mathcal M}^\varepsilon(\mathcal F_1\cup\mathcal F_2)$ is by
definition the smallest achievable value of $\max_{\gamma\in\mathcal
S}|\gamma|$ over $\varepsilon$-typical sufficient sets, and $\mathcal S$ is
one explicit feasible choice,
\begin{equation}
\ell_{\mathcal M}^{\varepsilon}(\mathcal F_1\cup\mathcal F_2)\le \max\big\{\ell_{\mathcal M}^{\varepsilon}(\mathcal F_1),\ell_{\mathcal M}^{\varepsilon}(\mathcal F_2)\big\}.
\end{equation}
This holds for \emph{every} mixing weight $p\in[0,1]$ and for \emph{any}
relationship between $\mathcal F_1$ and $\mathcal F_2$ (including
overlapping supports), since $\mu$ is constructed directly from $\mu_1,\mu_2$
as independent sampling sources rather than derived from a common ambient
measure.

\section{3. GEL and GEN beyond bipartition-extremization: marginal reconstructibility}
\label{subsec:reconstructibility}

Not every physically relevant GME measure is naturally
written in the bipartition-extremized form of Eq.~(1). We
now introduce a parallel notion -- \emph{marginal reconstructibility} --
that extends GEL and GEN to measures of this broader type.

\textbf{Definition 4.} (Marginally reconstructible GME measure)
\textit{Let $\mathcal M(\rho)$ be a GME measure and let $\mathcal F$ be a family of
$N$-qubit states (pure or mixed). For a subset $A\subseteq\{1,\dots,N\}$,
write $\rho_A$ for the corresponding reduced state, and let
$\mathcal P_N := \{A : A\subseteq\{1,\dots,N\},\ A\neq\varnothing\}$ denote
the set of all nonempty qubit subsets, with $|A|$ ranging from $1$ to $N$.
A subset $A\in\mathcal P_N$ \emph{reconstructs} $\mathcal M$ on a set of
states $\mathcal G\subseteq\mathcal F$ if there exists a function $h_A$,
depending only on $\rho_A$, such that
\begin{equation}
\mathcal M(\rho) = h_A(\rho_A), \qquad \forall\,\rho\in\mathcal G.
\label{eq:reconstruct}
\end{equation}
A subset $\mathcal S\subseteq\mathcal P_N$ is \emph{$\mathcal M$-sufficient
for $\mathcal F$ in the reconstructive sense} if, for every $\rho\in
\mathcal F$, there exists $A\in\mathcal S$ that reconstructs $\mathcal M$
on some $\mathcal G\subseteq\mathcal F$ with $\rho\in\mathcal G$ -- i.e.\
$\mathcal S$, together with a fixed choice of $\{h_A\}_{A\in\mathcal S}$,
covers $\mathcal F$. The set $\mathcal S = \{\{1,\dots,N\}\}$, with
$h_{\{1,\dots,N\}}:=\mathcal M$, is always $\mathcal M$-sufficient for any
$\mathcal F$, so a sufficient set exists trivially, exactly as
$\mathcal S=\mathcal B_N$ does in the bipartition-extremized case.}
\textbf{Definition 5.} (GME Evaluation Length and Number for marginally reconstructible measures)
\textit{Let $\mathcal M(\rho)$ be a marginally reconstructible GME measure and $\mathcal F$ a family of
$N$-qubit states. Among all $\mathcal M$-sufficient sets $\mathcal
S\subseteq\mathcal P_N$ for $\mathcal F$ in the sense of
Definition~4, the \emph{minimally demanding
sufficient sets} minimize the size of the largest subset they contain and,
subject to that, their cardinality:
\begin{equation}
\mathcal{S}^{*}(\mathcal{F})
=
\argmin_{\mathcal{S}\,:\,\mathcal{M}\text{-sufficient for }\mathcal{F}}
\left(
\max_{A\in\mathcal{S}} |A|,
\, |\mathcal{S}|
\right),
\label{eq:min_dem_suf_set_recon}
\end{equation}
with the minimization lexicographic as before. The GEL and GEN for
$\mathcal F$ are then
\begin{equation}
  \KM(\mathcal{F}) = \max_{A \in \mathcal{S}_i(\mathcal{F})}|A|, \qquad
  \NM(\mathcal{F}) = |\mathcal{S}_i(\mathcal{F})|,
  \qquad \mathcal S_i(\mathcal F)\in\mathcal S^*(\mathcal F).
\end{equation}
Since a sufficient set always exists,
$\KM(\mathcal F)$ and $\NM(\mathcal F)$ are always well-defined, with range
\begin{equation}
1\leq \KM(\mathcal{F}) \leq N, \qquad
1\leq \NM(\mathcal{F})\leq \sum_{k=1}^{\KM(\mathcal{F})}\binom{N}{k}\leq (2^{N}-1).
\end{equation}}

The upper bound on $\KM$ is now $N$, rather than $\lfloor N/2\rfloor$ as in
the bipartition-extremized case: unlike an extremization over $\mathcal
B_N$, which by construction never requires more than a proper subsystem,
reconstructibility of a general $\mathcal M$ may genuinely require the full
system, $\KM(\mathcal F)=N$, when no proper marginal determines $\mathcal
M(\rho)$ for every $\rho\in\mathcal F$ -- for instance, whenever $\mathcal F$
contains states that agree on every proper marginal but differ in
$\mathcal M$. The operational interpretation of the extremal values is
otherwise unchanged: $\KM(\mathcal F)=\NM(\mathcal F)=1$ means a single
single-qubit marginal already reconstructs $\mathcal M$ for every state in
$\mathcal F$; $\KM(\mathcal F)=N$ signals that $\mathcal M$, for this
family, is a genuinely global (non-marginal-reconstructible-below-$N$)
property, requiring full-state tomography.

Definitions~1 and 4 coincide
whenever both apply: if $\mathcal M$ is bipartition-extremized as in
Eq.~(1) and the extremizing bipartition $\gamma^*$ is the
same for every $\rho\in\mathcal F$, then $\gamma^*$ is trivially
reconstructing with $h_{\gamma^*}:=f_{\gamma^*}$, so the two definitions
agree and yield identical $\KM(\mathcal F),\NM(\mathcal F)$. Marginal
reconstructibility should thus be read as a strict generalization of the
sufficiency criterion underlying GEL and GEN and not as a distinct quantity, extending it to measures with no natural bipartition-extremized form and to mixed-state families.

\subsubsection{Example: amplitude-damped generalized GHZ states}

Consider the $N$-qubit generalized GHZ state,
\begin{equation}
\ket{g\mathrm{GHZ}_N} = a\ket0^{\otimes N} + b\ket1^{\otimes N}, \qquad |a|^2+|b|^2=1,
\end{equation}
subjected to single-qubit amplitude damping on qubit $1$ with damping
parameter $\gamma\in[0,1]$, described by Kraus operators
\begin{equation}
K_0 = \begin{pmatrix}1 & 0\\ 0 & \sqrt{1-\gamma}\end{pmatrix}, \qquad
K_1 = \begin{pmatrix}0 & \sqrt\gamma\\ 0 & 0\end{pmatrix}.
\end{equation}
This yields the family
\begin{equation}
\mathcal F_{\rm ad} = \Big\{\rho(a,\gamma) : |a|^2\in[0,1],\ \gamma\in[0,1]\Big\},
\end{equation}
with
\begin{equation}
\rho(a,\gamma) = |a|^2\ket{0^N}\!\bra{0^N} + (1-\gamma)|b|^2\ket{1^N}\!\bra{1^N}
+\gamma|b|^2\ket{01^{N-1}}\!\bra{01^{N-1}}
+\sqrt{1-\gamma}\big(ab^*\ket{0^N}\!\bra{1^N}+\mathrm{h.c.}\big).
\end{equation}
Since $\rho(a,\gamma)$ is an $N$-qubit X-state, its Genuine Multipartite
Concurrence is given exactly, via the formula of Ma \emph{et al.} (proven
for X-states by Hashemi Rafsanjani \emph{et al.}), by
\begin{equation}
C_g\big(\rho(a,\gamma)\big) = 2|ab|\sqrt{1-\gamma}.
\label{eq:CG_exact}
\end{equation}
We now determine $\KM(\mathcal F_{\rm ad})$ and $\NM(\mathcal F_{\rm ad})$
for $\mathcal M=C_g$.

\paragraph{No single-qubit bipartition is reconstructing.}
Tracing out all but the damped qubit gives
\begin{equation}
\rho_1 = \begin{pmatrix} 1-q & 0 \\ 0 & q\end{pmatrix}, \qquad q = (1-\gamma)|b|^2,
\end{equation}
which depends on $(a,\gamma)$ only through the single combination $q$.
Since $C_g$ in Eq.~\eqref{eq:CG_exact} depends on $|a|^2$ and $\gamma$
\emph{separately} rather than only through $q$, no function $h_{\gamma_1}$
of $\rho_1$ alone can reproduce $C_g$ for all $\rho\in\mathcal F_{\rm ad}$.
Explicitly, the pairs $(|a|^2,\gamma)=(0.8,0)$ and $(0.6,0.5)$ both give
$q=0.2$, hence identical $\rho_1=\mathrm{diag}(0.8,0.2)$, yet
\begin{equation}
C_g(0.8,0) = 2\sqrt{0.8\times0.2} = 0.8,
\qquad
C_g(0.6,0.5) = 2\sqrt{0.6\times0.4\times0.5} \approx 0.693,
\end{equation}
two different values -- so no reconstruction function $h_{\gamma_1}$ can
exist. For any undamped qubit $i\neq1$, the reduced state
$\rho_i=\mathrm{diag}(|a|^2,|b|^2)$ is entirely independent of $\gamma$,
while $C_g$ depends explicitly on $\gamma$ through $\sqrt{1-\gamma}$, so
$\rho_i$ fails even more directly. Hence no bipartition of size $1$ is
reconstructing for $\mathcal F_{\rm ad}$, and $\KM(\mathcal F_{\rm ad})>1$.

\paragraph{The two-qubit bipartition $\{1,2\}$ is reconstructing.}
Tracing out all but qubits $1,2$ gives
\begin{equation}
\rho_{12} = |a|^2\ket{00}\!\bra{00} + \gamma|b|^2\ket{01}\!\bra{01} + (1-\gamma)|b|^2\ket{11}\!\bra{11},
\end{equation}
with diagonal populations
\begin{equation}
P_{00}=|a|^2, \qquad P_{01}=\gamma|b|^2, \qquad P_{11}=(1-\gamma)|b|^2.
\end{equation}
Since $P_{00}+P_{01}+P_{11}=1$, these three populations carry two
independent parameters, exactly matching $(|a|^2,\gamma)$; inverting,
\begin{equation}
|a|^2 = P_{00}, \qquad |b|^2 = P_{01}+P_{11}, \qquad 1-\gamma = \frac{P_{11}}{P_{01}+P_{11}}.
\end{equation}
Substituting into Eq.~\eqref{eq:CG_exact},
\begin{equation}
C_g = 2\sqrt{P_{00}(P_{01}+P_{11})}\,\sqrt{\frac{P_{11}}{P_{01}+P_{11}}}
= 2\sqrt{P_{00}\,P_{11}},
\end{equation}
i.e.
\begin{equation}
C_g\big(\rho(a,\gamma)\big) = 2\sqrt{\rho_{00,00}\,\rho_{11,11}}
=: h_{\{1,2\}}(\rho_{12}), \qquad \forall\,\rho\in\mathcal F_{\rm ad},
\label{eq:CG_from_rho12}
\end{equation}
where $\rho_{00,00},\rho_{11,11}$ denote the corresponding diagonal
elements of $\rho_{12}$. Equation~\eqref{eq:CG_from_rho12} is precisely the
reconstruction condition~\eqref{eq:reconstruct} for $\gamma=\{1,2\}$, so
$\{1,2\}$ is a reconstructing bipartition, and the singleton set
$\mathcal S=\{\{1,2\}\}$ is $C_g$-sufficient for $\mathcal F_{\rm ad}$ in
the sense of Definition~4.

\paragraph{GEL and GEN.}
Combining the two results, the minimally demanding sufficient set for
$\mathcal F_{\rm ad}$ uses a bipartition of size $2$ and requires only one
such bipartition, giving
\begin{equation}
\KM(\mathcal F_{\rm ad}) = 2, \qquad \NM(\mathcal F_{\rm ad}) = 1.
\end{equation}
This example illustrates that even when $\mathcal M$ is not posed in the
bipartition-extremized form of Eq.~(1), GEL and GEN remain
well-defined and answer the same operational question: a symmetric pure
gGHZ family requires only single-qubit measurements ($\KM=1$), but
asymmetric local noise (damping on one identified qubit) raises the
requirement to a specific two-qubit joint measurement ($\KM=2$), even
though no single-qubit measurement, on any qubit, suffices.

\section{4. Proof of Proposition 1}

\textit{Proposition 1.} ~ Let $\ket\psi\sim\mu_{\mathrm Haar}$ be drawn from the Haar measure on $N$
qubits, and for $\gamma\in\mathcal B_N$ with size $|\gamma|$ let
\(
f_\gamma(\ket\psi) = C_{\gamma}(\ket\psi) = \sqrt{2\big(1-\Tr\rho_A^2\big)}
\)
denote the bipartite concurrence across the cut $\gamma$. There
exists a sequence $\varepsilon_N\to0$, decaying exponentially in $N$ ($\varepsilon_N\sim \mathcal{O}(2^{-N})$), such
that
\begin{equation}
\Pr_{\ket\psi\sim\mu_{\mathrm H}}\!\left[
\underset{\gamma\,:\,|\gamma|=1}{\mathrm{min}}\ C_{\gamma}(\psi) = \ C_{g}(\psi)
\right]
\;\ge\; 1-\varepsilon_N,
\label{eq:prop1}
\end{equation}
i.e.\ with probability at least $1-\varepsilon_N$ the minimization of
$C_\gamma(\psi)$ over $\mathcal B_N$ is attained already at bipartition
size $|\gamma|=1$.

\vspace{0.5cm}
To prove this we need to show $\Pr[P_{\gamma_1} > P_{\gamma_k}] \geq
1 - O(2^{-N})$ for all fixed $k>1$, where $\gamma_1$ denotes any
bipartition of size $1$ and $\gamma_k$ any bipartition of size $k$, with
means $\overline P_1,\overline P_k$ and variances $\sigma_1^2,\sigma_k^2$
as defined previously. For this we will use the following lemma.

\begin{lemma}
\label{lem:chebyshev}
Let $X$ and $Y$ be real random variables with means $m_X > m_Y$ and
variances $\sigma_X^2$ and $\sigma_Y^2$, respectively. Then
\begin{equation}
  \Pr[X > Y]
  \;\geq\;
  1 - \frac{4(\sigma_X^2 + \sigma_Y^2)}{(m_X - m_Y)^2}.
  \label{eq:lemma}
\end{equation}
\end{lemma}
\begin{proof}
The idea is simple: if $X$ and $Y$ have well-separated means, then the event
$X \le Y$ can only occur when at least one of the two variables deviates
substantially from its own mean. We now make this precise.

Define the midpoint
\begin{equation}
M=\frac{m_X+m_Y}{2}.
\end{equation}
If $X \le Y$, then it cannot be the case that both $X > M$ and $Y < M$.
Therefore, whenever $X \le Y$ occurs, at least one of the following two
events must also occur: $X \le M$ or $Y \ge M$. Hence,
\begin{equation}
\Pr[X \le Y] \le \Pr[X \le M] + \Pr[Y \ge M].
\end{equation}
We bound these two probabilities separately using Chebyshev's inequality.
Since
\begin{equation}
M-m_X = -\frac{m_X-m_Y}{2},
\end{equation}
the event $X \le M$ implies that $X$ deviates from its mean by at least
$(m_X-m_Y)/2$. Thus,
\begin{align}
\Pr[X \le M]
&= \Pr\!\left(X-m_X \le -\frac{m_X-m_Y}{2}\right) \nonumber\\
&\le \Pr\!\left(\bigl|X-m_X\bigr| \ge \frac{m_X-m_Y}{2}\right) \nonumber\\
&\le \frac{4\sigma_X^2}{(m_X-m_Y)^2}.
\end{align}
Similarly, since
\begin{equation}
M-m_Y = \frac{m_X-m_Y}{2},
\end{equation}
the event $Y \ge M$ implies that $Y$ deviates from its mean by at least
$(m_X-m_Y)/2$, and therefore
\begin{equation}
\Pr[Y \ge M]
\le \Pr\!\left(\bigl|Y-m_Y\bigr| \ge \frac{m_X-m_Y}{2}\right)
\le \frac{4\sigma_Y^2}{(m_X-m_Y)^2}.
\end{equation}
Combining the two estimates gives
\begin{equation}
\Pr[X \le Y]
\le \frac{4(\sigma_X^2+\sigma_Y^2)}{(m_X-m_Y)^2}.
\end{equation}
Taking the complement yields
\begin{equation}
\Pr[X > Y]
\ge 1 - \frac{4(\sigma_X^2+\sigma_Y^2)}{(m_X-m_Y)^2},
\end{equation}
which is the claimed bound.
\end{proof}


\begin{figure}
    \centering
    \includegraphics[width=0.65\linewidth]{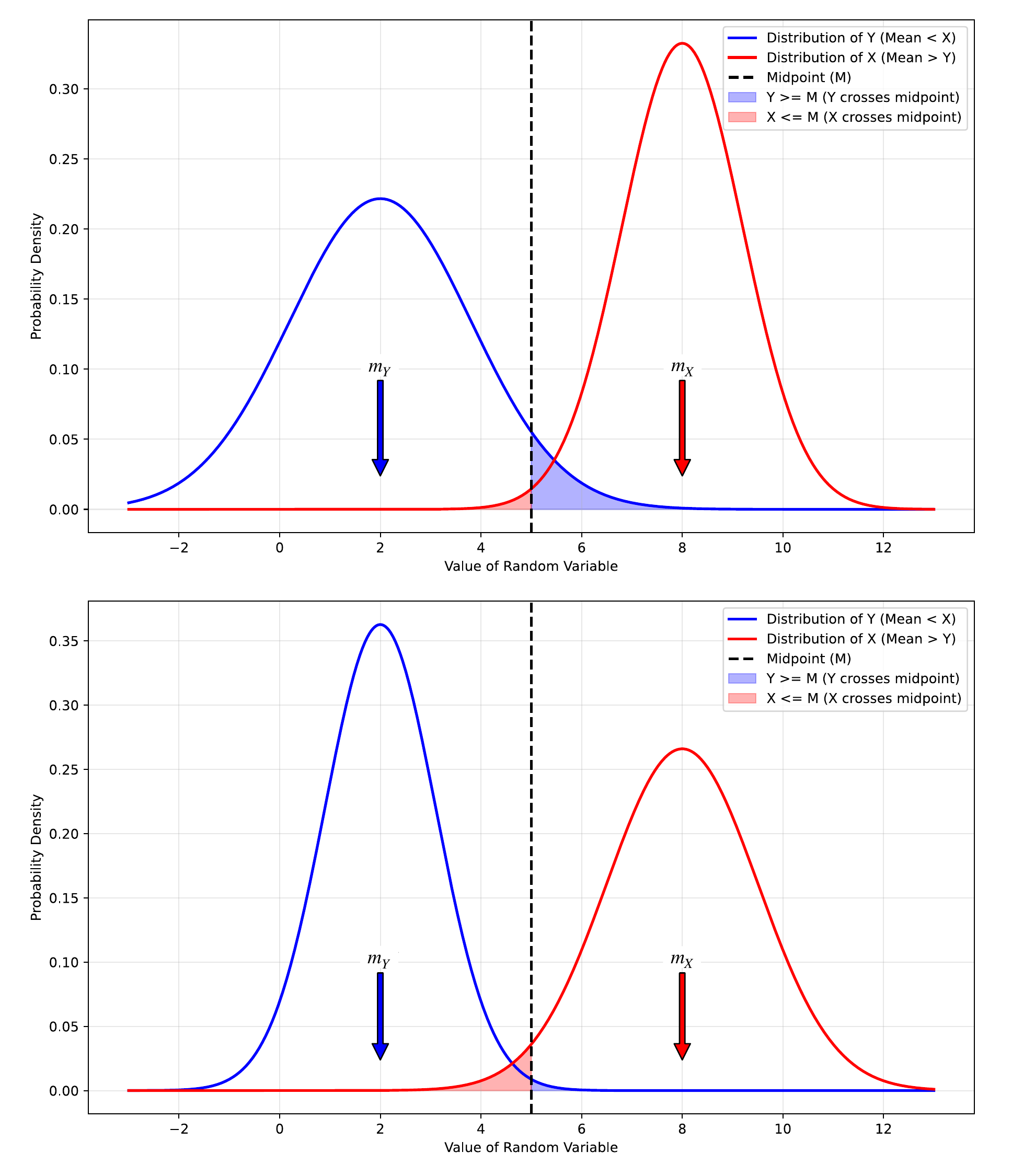}
    \caption{\textbf{Visualization of the Upper Bound for $\Pr[X \le Y]$.} The solid curves represent the probability density functions for two random variables, $Y$ (left) and $X$ (right), with respective means $m_Y$ and $m_X$ such that $m_X > m_Y$. The dashed vertical line denotes the midpoint barrier, $M = (m_X + m_Y)/2$. For the event $X \le Y$ to occur, a necessary condition is that at least one of the variables must cross this midpoint barrier. The shaded regions illustrate these failure modes: the shaded area under $Y$ represents $\Pr[Y \ge M]$, and the shaded area under $X$ represents $\Pr[X \le M]$. Bounding the probability of $X \le Y$ by the sum of these two shaded tail areas, $\Pr[X \le M] + \Pr[Y \ge M]$, establishes an upper bound. It is an upper bound, rather than an exact probability, because it inherently over counts the joint parameter space where the variable cross the midpoint barrier but do not satisfy  $X \le Y$. This midpoint-barrier construction enables the application of Chebyshev's inequality, allowing the resulting bound to be expressed solely in terms of the mean and variance of the associated random variable.}
    \label{fig:midpoint_upper_bound}
\end{figure}



\subsection{Proof of Proposition~1}
The proof follows by showing that, for fixed subsystem size $k>1$, the
single-qubit purity has a strictly larger mean than $P_{\gamma_k}$ while
both fluctuations are exponentially suppressed in $N$.

\begin{proof}
We apply Lemma~\ref{lem:chebyshev} to the subsystem purities
$P_{\gamma_k} := \Tr\rho_A^2$ for an $N$-qubit Haar-random pure state
$\ket\psi\sim\mu_{\mathrm H}$, with $|\gamma_k|=k$, $d_A=2^k$, $d_B=2^{N-k}$,
and mean and variance $\overline P_k := \mathbb E[P_{\gamma_k}]$,
$\sigma_k^2 := \mathrm{Var}(P_{\gamma_k})$, as introduced above. For any
$k>1$, taking $X=P_{\gamma_1}$ and $Y=P_{\gamma_k}$ in
Lemma~\ref{lem:chebyshev} gives
\begin{equation}
\Pr\big[P_{\gamma_1} > P_{\gamma_k}\big] \;\geq\; 1 - \frac{4(\sigma_1^2 + \sigma_k^2)}{(\overline P_1 - \overline P_k)^2}.
\label{eq:puritydom}
\end{equation}
To establish that this bound approaches unity as $N$ grows, we need to show
that the ratio $(\sigma_1^2+\sigma_k^2)/(\overline P_1-\overline P_k)^2$
vanishes with increasing system size. We verify this by computing each
factor explicitly.

\smallskip
\noindent\textit{Mean separation.}
\begin{equation}
\overline P_1 = \frac{2 + 2^{N-1}}{2^N + 1},
\qquad
\overline P_k = \frac{2^k + 2^{N-k}}{2^N + 1}.
\end{equation}
Their difference is
\begin{equation}
\overline P_1 - \overline P_k = \frac{2^{N-1} - 2^{N-k} + 2 - 2^k}{2^N + 1}.
\end{equation}
For fixed $k\geq2$ and large $N$,
\begin{equation}
\overline P_1 - \overline P_k \approx \frac{2^{N-k}(2^{k-1}-1)}{2^N} = \frac{2^{k-1}-1}{2^k} = O(1).
\end{equation}

\smallskip
\noindent\textit{Variance suppression.}
\begin{equation}
\sigma_1^2
=
\frac{
2(2^2-1)(2^{2(N-1)}-1)
}{
(2^N+1)^2(2^N+2)(2^N+3)
}
\approx \frac{6}{2^{2N+2}}.
\end{equation}

For a fixed $k$-qubit subsystem, $m=2^k$ and $n=2^{N-k}$, giving
\begin{equation}
\sigma_k^2
=
\frac{
2(2^{2k}-1)(2^{2(N-k)}-1)
}{
(2^N+1)^2(2^N+2)(2^N+3)
}
\approx \frac{2}{2^{2N}}.
\end{equation}

Thus,
\begin{equation}
\sigma_1^2+\sigma_k^2
=
O(2^{-2N}).
\end{equation}

\smallskip
\noindent\textit{Combining the bounds.}
Since $(\overline P_1 - \overline P_k)^2 = O(1)$ and
$\sigma_1^2+\sigma_k^2 = O(2^{-2N})$,
\begin{equation}
\frac{4(\sigma_1^2 + \sigma_k^2)}{(\overline P_1 - \overline P_k)^2} = O(2^{-2N}),
\end{equation}
and therefore $\Pr[P_{\gamma_1} > P_{\gamma_k}] \geq 1 - O(2^{-2N})$ for all
fixed $k>1$. While $P_{\gamma_1}$ and $P_{\gamma_k}$ are correlated, the
union bound applied to the midpoint barrier in the proof of
Lemma~\ref{lem:chebyshev} circumvents the need for independence, and the
suppression of the variance guarantees the bound remains tight. Since
$C_\gamma$ is a monotonically decreasing function of $P_\gamma$, this
directly establishes $\Pr[C_{\gamma_1} < C_{\gamma_k}] \geq 1-O(2^{-2N})$
for all fixed $k>1$; a union bound over the (at most $2^{N-1}$) bipartitions $\gamma'\in\mathcal
B_N$ with $|\gamma'|>1$ then yields Eq.~\eqref{eq:prop1} with
$\varepsilon_N = O(2^{N-1})\cdot O(2^{-2N}) = O(2^{-N})$.
\end{proof}

\begin{figure}
    \centering
    \includegraphics[width=0.85\linewidth]{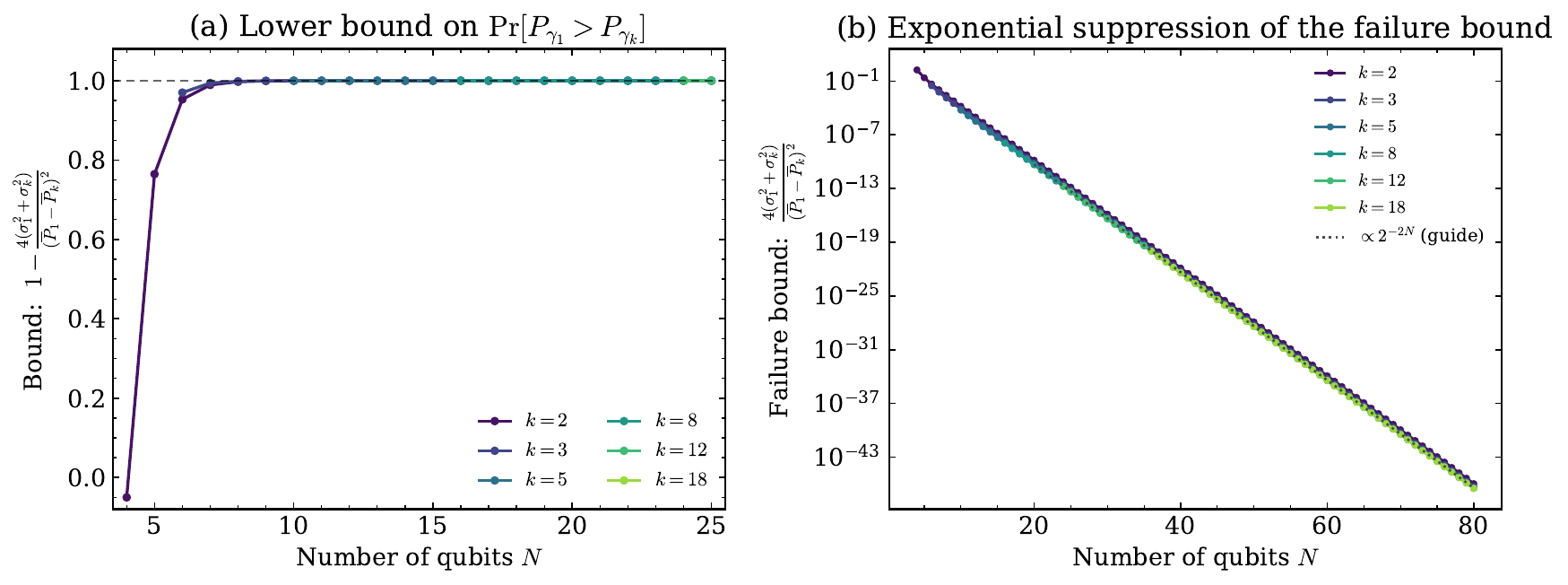}
    \caption{Chebyshev bound on single-qubit purity domination for Haar-random
$N$-qubit pure states, evaluated exactly from the Haar purity moments
$\overline P_k$, $\sigma_k^2$ for representative bipartition sizes
$k=2,3,5,8,12,18$. \textbf{(a)} The lower bound
$1-4(\sigma_1^2+\sigma_k^2)/(\overline P_1-\overline P_k)^2$ on
$\Pr[P_{\gamma_1}>P_{\gamma_k}]$ [Eq.~\eqref{eq:puritydom}] rises sharply
from small values at $N\lesssim5$ to unity by $N\approx10$, confirming that the single-qubit purity exceeds that of any larger bipartition with probability approaching $1$ already at modest system sizes. \textbf{(b)} The corresponding failure
bound $4(\sigma_1^2+\sigma_k^2)/(\overline P_1-\overline P_k)^2$ on a logarithmic scale, decays exponentially in $N$ and closely tracks the
$\propto2^{-2N}$ reference (dotted), consistent with the per-pair rate established in Lemma~\ref{lem:chebyshev}.}
\label{fig:puritybound}
    \label{fig:purity_dom}
\end{figure}


\section{5. Reduction to the most unbalanced cut and expansion of the GMC}
\label{subsec:3C}
We established in Proposition~1 that for Haar-random pure
states, the minimization over bipartitions in the definition of $C_g$ is
achieved, with probability $1-\varepsilon_N=1-O(2^{-N})$, by a single-qubit
bipartition. It therefore suffices to consider a fixed cut
$\gamma_1=(A^*,\bar A^*)$ with $|\gamma_1|=1$. Let
\begin{equation}
P_{\gamma_1} := \Tr\rho_{A^*}^2
\end{equation}
denote the associated subsystem purity. The GMC then reduces to
\begin{equation}
C_g = \sqrt{2\bigl(1 - P_{\gamma_1}\bigr)},
\label{eq:g-def}
\end{equation}
and its statistical properties are determined entirely by the distribution
of $P_{\gamma_1}$ under the Haar measure $\mu_{\mathrm H}$.

\paragraph*{Expected value of the GMC.}
To evaluate $\mathbb E[C_g]$, we expand the square root around the mean
purity $\overline P_1 := \mathbb E[P_{\gamma_1}]$. Writing
\begin{equation}
P_{\gamma_1} = \overline P_1 + \delta,
\end{equation}
where $\delta = P_{\gamma_1}-\overline P_1$ is a zero-mean fluctuation, we
expand
\begin{equation}
C_g = \sqrt{2(1 - \overline P_1 - \delta)}
= \sqrt{2(1-\overline P_1)}\left(1 - \frac{\delta}{2(1-\overline P_1)}\right)^{1/2}.
\end{equation}
Expanding to second order in $\delta$ and taking expectations, using
$\mathbb E[\delta]=0$ and $\mathbb E[\delta^2]=\sigma_1^2$, gives
\begin{equation}
\mathbb{E}[C_g]
\approx \sqrt{2(1-\overline P_1)}
- \frac{\sigma_1^2}{4\sqrt{2}(1-\overline P_1)^{3/2}}.
\label{Eq:expected Cg}
\end{equation}
The first term is the leading contribution, determined by the mean purity alone, while the second term captures the correction due to fluctuations. Since $\sigma_1^2=O(2^{-2N})$ (established in the proof of Proposition~1, this correction is exponentially suppressed in system size.

\paragraph*{Asymptotic behavior.}
For the single-qubit bipartition, the mean purity is
\begin{equation}
\overline P_1 = \frac{2 + 2^{N-1}}{2^N + 1} \xrightarrow{N \to \infty} \frac{1}{2}.
\end{equation}
Substituting into the leading term of Eq.~\eqref{Eq:expected Cg}, $\mathbb{E}[C_g] \to 1$ as $N \to \infty$, so the average GMC approaches its maximum possible value of unity in the thermodynamic limit.

\paragraph*{Finite-$N$ behavior and rate of convergence.}
For finite $N$, Eq.~\eqref{Eq:expected Cg} provides a quantitatively accurate approximation to $\mathbb E[C_g]$. To assess how quickly the asymptotic value is approached, we note that the leading correction to
$\overline P_1$ from its asymptotic value of $1/2$ is
\begin{equation}
\overline P_1 - \frac{1}{2}
= \frac{2 + 2^{N-1}}{2^N + 1} - \frac{1}{2}
= \frac{3}{2(2^N+1)}
\approx \frac{3}{2^{N+1}},
\end{equation}
which decays exponentially in $N$. The leading term of $\mathbb E[C_g]$
therefore approaches unity exponentially fast in $N$.

\paragraph*{Higher-order corrections.}
For completeness, we note that higher-order terms in the expansion of
$\mathbb E[C_g]$ involve central moments
\begin{equation}
\kappa_r = \mathbb{E}\big[(P_{\gamma_1} - \overline P_1)^r\big], \quad r \geq 3.
\end{equation}
These are suppressed even faster, following from the joint eigenvalue
distribution and the moment formulae of Giraud~\cite{Giraud2007}. They
contribute negligibly for any $N$ large enough that the second-order
correction is already small, and we do not evaluate them explicitly here.

\section{6. Local Random Circuits}

\begin{figure*}
    \centering
    \includegraphics[width=\linewidth]{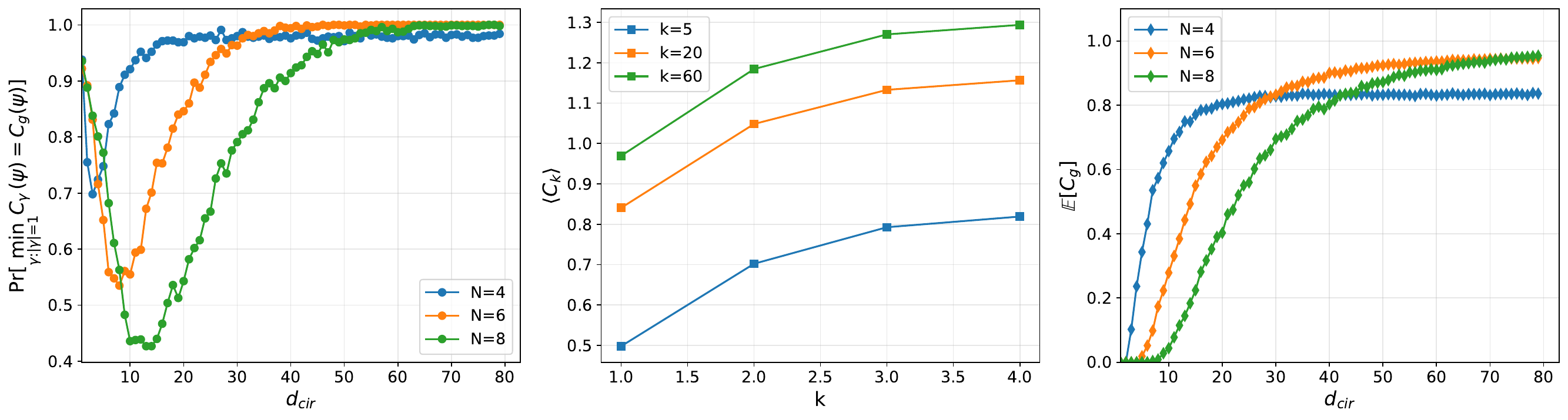}
    \caption{Behavior of the concurrence-based GME under local random circuits. 
(a) Probability that the minimization defining $C_g$ is achieved by the most unbalanced bipartition ($k=1$) as a function of the number of gates. 
(b) Bipartite concurrences as a function of subsystem size $k$ for different numbers of gates, showing a persistent ordering in which the most unbalanced bipartition yields the smallest value. 
(c) Average $\mathbb{E}[C_g]$ as a function of the number of gates for different system sizes $N$. 
All quantities are averaged over circuit realizations.
}
    \label{fig:cgmevsdepth}
\end{figure*}

Figure~2(a): The probability that the minimization defining $C_g$ is attained by the most unbalanced bipartition ($|\gamma|=1$) increases rapidly with $d_{\rm cir}$ and converges to unity for all system sizes considered. At small depths several cuts can compete due to the near absence of GME, but this ambiguity disappears once scrambling develops. The minimization then becomes overwhelmingly dominated by the single-qubit cut.

The origin of this behavior is shown in Fig.~2(b), where the average concurrence associated with each bipartition size $k$ is plotted for $N=8$ and representative depths $d_{\rm cir}\in\{5,20,60\}$. Throughout the evolution, more balanced bipartitions exhibit larger average concurrence than more unbalanced ones, in agreement with the ordering of mean purities. Although all concurrences increase with depth, the ordering remains intact, with the $|\gamma|=1$ bipartition consistently yielding the smallest concurrence. Consequently, the GMC minimization collapses onto the most unbalanced cut.

Figure~2(c) shows the corresponding evolution of $\mathbb{E}[C_g]$. The GMC grows from zero, rapidly approaches the Haar-typical value, and subsequently saturates. Once the $|\gamma|=1$ cut dominates the minimization, further growth of multipartite correlations at larger scales has little influence on the value of $C_g$.

\end{document}